\documentclass[lettersize,journal]{IEEEtran}
\usepackage{amsthm}
\usepackage{hyperref}
\usepackage[utf8]{inputenc}
\usepackage[english]{babel}
\usepackage{amsfonts}
\usepackage{verbatim}
\usepackage{mathtools, cuted}
\usepackage{adjustbox}
\usepackage{amsmath}
\usepackage{amssymb}
\usepackage{graphicx}
\usepackage{graphics}
\usepackage{tabularx}
\usepackage{color}

\usepackage[caption=false,font=footnotesize]{subfig}

\usepackage{cite}
\usepackage{lipsum}
\usepackage[ruled,linesnumbered,vlined]{algorithm2e}
\usepackage[compact]{titlesec}

\usepackage{physics}

\newcommand*\tageq{\refstepcounter{equation}\tag{\theequation}}

\renewcommand{\footnoterule}{%
  \kern -3pt
  \hrule width 0.5\linewidth
  \kern 2.6pt
}

\usepackage{tikz}

\def\BibTeX{{\rm B\kern-.05em{\sc i\kern-.025em b}\kern-.08em
    T\kern-.1667em\lower.7ex\hbox{E}\kern-.125emX}}

\usepackage{accents}

\newtheorem{remark}{Remark}
\newtheorem{theorem}{Theorem}
\newtheorem{corollary}{Corollary}
\newtheorem{lemma}{Lemma}

\newtheorem{utility}{Utility}

\newcommand{\one}   {{\mathds{1}}}

\newcommand{\R}     {\mathbb{R}} 

\newcommand{\N}     {\mathbb{N}}
 
\renewcommand{\P}   {\mathbb{P}} 
 
\newcommand{\E}     {\mathbb{E}}

\newcommand{\Pcal}   {{\mathcal P }} 
 
\newcommand{\Rcal}   {{\mathcal R }}

\newcommand{\1}   {{\mathbbm 1 }} 
\newcommand{\nn}   {{\nonumber}} 
\newcommand{\half}   {{\frac{1}{2}}}

\renewcommand\thesubfigure{\arabic{subfigure}}

\def\one{\mathbbmss{1}}

\def\a{\alpha}

\def\la{\lambda}

\def\t{{\tau}}

\usepackage{bbm} 

\newenvironment{multiparacolor_red}{\par\color{red}}{\par} 

\newenvironment{multiparacolor_magenta}{\par\color{magenta}}{\par}

\title{Doppler-Shannon Association in Vehicular Networks: A Stochastic Geometry Analysis}

\author{Ashutosh Balakrishnan, François Baccelli, Sanjoy Kumar Jhawar, and Philippe Martins \thanks{A. Balakrishnan and P. Martins are with the Department of Computer Science and Networks, Télécom Paris, Palaiseau, France (email: \{balakrishnan, martins\}@telecom-paris.fr). F. Baccelli is with ENS Paris and Centre Inria de Paris (email: francois.baccelli@inria.fr). S. K. Jhawar is with Centre Inria de Paris and Télécom Paris, France (email: sanjoy-kumar.jhawar@inria.fr, sanjoy.jhawar@telecom-paris.fr).  
} 
}

\begin{document}
\maketitle

\makeatletter
\renewcommand{\thesubfigure}{\alph{subfigure}}
\renewcommand{\p@subfigure}{}
\makeatother

\begin{abstract}
Historically, in cellular networks, the user equipment (UE) association with a base station (BS) has been largely confined to nearest BS association strategies. 
Unlike cellular networks wherein the BSs are static, vehicular networks are prone to high mobility, resulting in the UE experiencing a Doppler shift. The Doppler shift   further results in the fading channel exhibiting high temporal variability. 
In this work, we present a novel stochastic geometry framework, modeling the effects of the Doppler shift on the physical layer network performance through the coherence time. First, the statistics of the Doppler shift are studied in a two- and three-dimensional Poisson point process (PPP) setting. The distribution and the geometric factors influencing the Doppler shift are analyzed. Next, we propose Doppler shift aware utilities, modeling the Doppler shift induced Shannon rate achievable in time-selective fading links. Through these utilities, we analytically study the outage and coverage probabilities for a UE to achieve a given quality of service. These probabilities are computed based on the optimal BS association in vehicular networks, which goes beyond the closest Euclidean distance. In the UE-BS line of sight, the novel BS association policies factor the BS direction of motion and the fading gain, in addition to the Euclidean distance. We show that the nearest BS is no longer the optimal BS association, resulting in non-convex coverage tessellations associated with the BSs. Our analytical and  simulation results illustrate the need for a Doppler shift correction, complementing the ideal SNR based Shannon performance. The network performance achieved through the proposed Doppler-Shannon based association is proved, analytically as well as through simulations, to be superior to the classical Shannon based association achieving spatially averaged gains up to 90\% and 22\% in a 2D and 3D setting, respectively. 

\end{abstract}



\section{Introduction}\label{secI}
Ubiquitous connectivity is a key theme of the upcoming sixth generation (6G) networks. Vehicular networks, e.g., high altitude platform stations (HAPS) and low Earth orbit (LEO) satellites are envisioned to be a key enabler of ubiquitous connectivity, providing coverage in areas devoid of terrestrial infrastructure \cite{v2v-csi}. 6G is also expected to support high mobility applications like high speed railways. Broadly, vehicular networks can be broadly defined as any network wherein the network nodes are mobile, rather than static. 
This mobility of nodes results in the user equipment (UE)  experiencing Doppler shift \cite{Doppler-1, Doppler-2}. Vehicular nodes like UAVs have speeds of around 50 m/s, while high speed trains and LEO satellites have much higher speeds up to 150 m/s and 2000 m/s, respectively. 

%
%
The Doppler shift in the frequency domain results in slow/ fast fading of the wireless channel in the time domain \cite{Doppler-poles}. Hence, high Doppler scenarios involve rapid variations of the wireless channel, which impact the channel statistics in time. This temporal variability of the fading channel is termed as \textit{time selectivity}, and should be factored in vehicular system design \cite{jsac}. 
In this work, we aim to design Doppler aware association schemes in large communication networks and we use  stochastic geometry to further evaluate the system performance incorporating different mobility scenarios.

In this work, the network performance  is evaluated using stochastic geometry tools. Through such tools, we have been able to determine the spatial statistics of the Doppler shift, in terms of {\em ensemble average} of the system, without the need for observing the system for a long period of time. The spatial \textit{ergodicity} of the underlying stochastic system is   formalized using random locations of the mobile nodes as a Poisson point process (PPP), in addition to i.i.d random directions of motion (DoM) of these points. The metrics of interest are designed based on functions of this geometric randomness, for instance, the angle formed by the line of sight (LoS) and the DoM and subsequently the coherence time which influences the Doppler shift induced Shannon rate. 

Classical studies on UE association and handovers have been based on 
 functions of the Euclidean distance, e.g., signal-to-noise-ratio (SNR), the Shannon rate. We term this classical association as the \textit{Shannon association policy}. In this work, we propose spatial averages factoring in the UE-BS geometry formed by the  directionality of the base station (BS), beyond the Shannon association geometry. We term it as the \textit{Doppler-Shannon association policy.} We believe that future communication system design (including but not limited to  UE-BS association, coverage tessellations, and handovers) in such vehicular wireless networks need to be Doppler shift aware. 
To the best of our knowledge, this is the first work to unify the directionality of the mobile BS, its speed, the Euclidean LoS from UE, and the fading gain of the channel in system-level design and performance evaluation of vehicular networks.

\subsection{Related works and Motivation}

From a stochastic geometry perspective, the authors in \cite{dhillon-2018} derive the coverage probability in vehicular networks, modeling them as a Poisson Line Process (PLP). Vehicular network performance under platooning is studied in \cite{gupta-dhillon-2023}. The distribution of loads served by a BS is derived by modeling roads as a PLP and the UEs as a Mat\'{e}rn cluster process. The authors in \cite{haenggi-2021} propose a transdimensional approach, superpositioning one- and two-dimensional PPPs. This approach is observed to provide better approximations as compared to a PLP based network modeling. 
A generalized framework for the constellation of satellites and its performance evaluation is studied in \cite{choi2024novel} using Poisson-Cox point process and in \cite{chang-sik-fb24} for the multi-altitude setting. The work in \cite{khan2020stochastic} presents a stochastic geometry framework to study the Doppler shift experienced by a user in a cluster. A recent work in \cite{al2024doppler} presents a stochastic geometry framework to analyze Doppler shift in  mega LEO satellite constellations. 

From a signal processing perspective, the authors in \cite{initial-access-doppler} study initial access protocol design in the presence of high Doppler shift. In a multi-path scenario, a Doppler shift estimation framework was proposed in  \cite{Doppler-est}. The frameworks in \cite{ali1999predicting, li2002analytical, ali2002doppler, vehicular-phy, papapetrou2005analytic} study LEO satellite networks and propose protocols for computing Doppler shift and use it for handoffs, and other network applications.

All the recent stochastic geometry based vehicular network models in \cite{dhillon-2018, gupta-dhillon-2023, haenggi-2021, dhillon-2024} have largely ignored the fact that, in vehicular networks, the base stations (BSs) (modeled through a point process) are in motion. The  frameworks in \cite{choi2024novel}, \cite{talgat2020stochastic}    compensate the Doppler shift by exploiting the available information on vehicular trajectory, or compensate the effect of the Doppler shift  through the channel fading gain. We feel that considering the effect of the Doppler shift through fading gain might be too simplistic, as it ignores the inherent stochastic geometry influencing the Doppler shift.   Further, the UE-BS association strategies have been largely confined to optimum SNR  reception strategies, which are based on Euclidean distances. The work in  \cite{dhillon-2024} goes beyond Euclidean distance by considering the angular distance as well, but the work is confined to a beam managament strategy from a BS perspective.



Vehicular communication channels, unlike terrestrial channels, 
are prone to highly mobile and dynamic environments. Hence, they exhibit temporal variability and non-stationary channel statistics \cite{vehicular-phy}.  These vehicular radio channels are both time as well as frequency selective \cite{v2v-fading}; hence there is a need for a unified framework capturing the effect of node mobility, node density, and the physical layer network performance. 

In this work, we devise a novel stochastic geometry framework for studying  the effect of vehicular motion through Doppler shift, and further capturing it in the physical layer system network performance through the  coherence time. \textit{A key concept which we aim to convey through this work is the fact that the classical  Shannon analysis is not enough to gauge the performance of vehicular networks.} For this reason, we propose a novel Doppler-Shannon association framework capturing the Doppler-induced Shannon achievable network performance, thereby shaping up the stochastic geometry version of  the SNR based Shannon theory to vehicular networks.   



The proposed Doppler-Shannon association is based on   utilities, which are designed as functions of the geometric angle formed between the LoS and the DoM of the BS from the UE, the wireless fading gain, and the Euclidean UE-BS distance.  
%
Based on these utilities, optimal BS association for the UE is performed, and we show the need to go beyond the traditional nearest BS association policy.
These novel Doppler-Shannon association frameworks are shown to result in non-convex coverage tessellations, replacing the nearest BS association policy which is  no longer optimal. 

The network performance is evaluated by modeling the BSs as a PPP. We utilize stochastic geometry tools to evaluate the performance of the utilities in a stochastic network, averaging the performance metric (outage or coverage   probability) over the geometric angle between UE LoS and BS DoM, in addition to the fading gain and Euclidean distance. Our simulation results confirm the need for a Doppler correction in vehicular networks to complement the classical Shannon association. 
We summarize the key contributions below.

\subsection{Contributions}

\begin{enumerate}
    \item First, in the two-dimensional (2D) and three-dimensional (3D) setting, we characterize the statistics of the Doppler shift experienced at a static UE receiver, from a mobile BS. The Doppler shift is characterized by statistically analyzing the radial component of the geometric angle between the LoS and DoM, along the LoS. 

    \item In the 2D setting, the Doppler shift  distribution is shown to be following an arc-sine  distribution, while in the 3D setting the Doppler shift is uniformly distributed. It is proved that the Doppler shift experienced by a static UE is independent of the DoM of the mobile BS. The rotation invariance of the statistical analysis is showcased for both 2D and 3D settings. 
    


    \item Next, we study the effect of Doppler shift on the network performance, by capturing the  Doppler shift induced Shannon rate achievable on a time selective wireless fading channel. 
    To capture the impact of Doppler spread on the Shannon capacity, we propose two novel utility functions  to analytically model: (a) the data transmitted per coherence time interval and (b) the  data rate.
    These proposed utilities  are  functions of the geometrical angles (capturing the Doppler shift and therein the coherence time), in addition to the fading gain, and the Euclidean distance between the UE and a spatial random BS. 
    \item Based on these Doppler aware utilities, we propose a \textit{Doppler-Shannon association} framework for the UE in large vehicular networks. It is shown that, due to the mobility of BSs, the nearest BS may no longer be the optimal.
    Rather, the association is performed by jointly capturing the geometric quantities: (a) the relative angle of motion of BS with respect to the UE termed as the ``Doppler angle'', (b) the Euclidean distance, and  (c) the channel fading gain.
    \item With mobile BS locations modeled as a PPP, stochastic geometry tools are used to compute the outage and success  probability given the user QoS. The UE QoS is further averaged, to visualize the effect of the density of points and velocity, on the average achievable network QoS. The coverage tessellations are observed to be more and more non-convex when increasing the Doppler shift. 

    \item Our results demonstrate that, unlike static cellular networks where classical Shannon performance evaluation is  sufficient, a Doppler correction is needed in vehicular networks to complement the Shannon performance.
    The Doppler-Shannon association  framework is shown to provide up to 90\% and 22\% gains, respectively, in the 2D and 3D PPP settings at a path-loss exponent (PLE) of $2.5$. In a more cluttered environment (PLE 4), we obtain  up to 14\% gains in a 2D setting, and up to 16\% gains in a 3D setting over the Shannon association framework.
    
    
    
\end{enumerate}





\section{System Model} \label{secII}

The cellular network model consists of mobile BSs modeled as a homogeneous PPP $\Pcal:=\sum_{i}\delta_{X_i}$ of intensity $\lambda$ on $\R^d$. Here, $\delta_{X_i}$ denotes a point mass at location $X_i$. In this work, we  analyze the system in the two- and three-dimensional scenarios, i.e., $d = 2, 3$. The system also comprises of a static UE, assumed to be located at the origin.

A BS randomly located at point $X_i$ of $\Pcal$, is assumed to have an i.i.d. DoM vector \underbar{$v$}  with  a constant norm $v$. 
%
%
%
%
%
The LoS vector from the UE to the BS is denoted by \underbar{$r$}, with norm $r$. The parameters characterizing the LoS and DoM depend on the level of abstraction, and will be discussed in detail for the  2D and 3D settings in the upcoming sub-sections. The BSs are also assumed to be transmitting radio signals of wavelength $\rho$. 
In general, the Doppler shift experienced at the UE is 
    \begin{equation} \label{doppler1}
        D = \frac{\underbar{$v$}_r}{\rho},
    \end{equation}
where $\underbar{$v$}_r$ denotes the radial component of velocity along the LoS. Let $\Gamma$ denote the relative angle of motion between the LoS and DoM vectors, influencing the Doppler shift (see Fig. \ref{fig:1} for a 2D setting). Then the radial component of velocity $\underbar{$v$}_r$ comes from the cosine component of the DoM, i.e.,  $\underbar{$v$}_r = v \cos \Gamma$.  Therefore, the Doppler shift experienced is further expressed as 
\begin{equation}\label{doppler2}
    D (\Gamma) = \frac{v}{\rho} \cos \Gamma.
\end{equation}
The angle $\Gamma$, between the LoS and DoM, influencing the Doppler shift will henceforth be termed as \textit{the Doppler angle}. The cosine component of the Doppler angle $\Gamma$,  made along the UE-BS LoS, is one of the geometric stochastic quantities of interest. 

In vehicular networks, the mobility of the BSs results in the wireless channels being specifically prone to time selectivity. In time, the wireless channel fading is classified as slow or fast fading, depending on the rate of change of channel statistics. In particular, a wireless time-selective channel is said to be statistically  ``flat''  within a \textit{coherence time} interval $T_c$. The coherence time is directly related to the Doppler spread (and hence the Doppler shift in the frequency domain) as
\begin{equation}
T_c (\Gamma) = \frac{\mu}{|D|} = \frac{\mu \rho}{v |\cos \Gamma|}.\nn
\end{equation}
Here, $\mu$ is a constant and $|D|$ is the magnitude of the Doppler shift. Higher mobility of the nodes may result in higher Doppler shift, depending on the relative Doppler angle, thereby reducing the coherence time.  
The Doppler  angle $\Gamma$  between the 
LoS and DoM
%
%
is no longer constant, but dynamically varying in time. Hence, the Doppler shift and thereby the coherence time interval at the UE from a random spatial BS is now random and varies over time, depending on spatial as well as the directional information about the BS. 

Below we derive the statistics of the Doppler angles for a 2D and 3D setting. 


\begin{figure}
    \centering
    \includegraphics[width=0.4\textwidth]{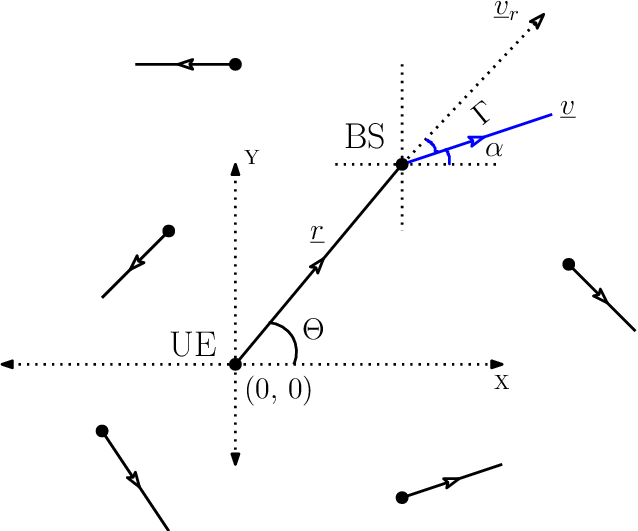}
    \caption{2D geometry of the vehicular network consisting of a  static UE receiver and mobile BSs.}
    \label{fig:1}
\end{figure}

\subsection{2D statistical study}

In a 2D scenario, for a BS located at a point $X$ of the  homogeneous PPP $\Pcal$, the LoS vector \underbar{$r$} between a typical UE and a BS is parameterized through the angle $\Theta \stackrel{d}{\sim} \mbox{U}[0, 2\pi)$ which it makes with the $x$-axis (Fig.~\ref{fig:1}). The $x$ and $y$ Cartesian coordinates of \underbar{$r$} are then $\underbar{$r$} {=}[r\cos\Theta\; r\sin\Theta]^T$.
Similarly, let the DoM of the BS with respect to the $x$-axis be $\alpha \stackrel{d}{\sim} \mbox{U}[0, 2\pi)$, having  Cartesian coordinates, $\underbar{$v$} {=} [v\cos \alpha\; v \sin \alpha]^T$. It may be noted that our analysis is based on determining the Doppler shift measured at a typical UE from a generic BS. 


The corresponding Doppler angle $\Gamma$ is computed using an inner product as 
\begin{equation}\label{2d-gamma}
\begin{aligned}
\Gamma {=} \cos^{-1}\left( \frac{ \langle\underbar{$r$}, \ \underbar{$v$}\rangle }{rv} \right) 
{=} \cos^{-1}( \cos(\Theta {-} \alpha)) {=} (\Theta {-} \alpha)\bmod 2\pi.
\end{aligned}
\end{equation}
In a 2D scenario, the resulting Doppler shift  experienced at the UE is 
\begin{equation}\label{doppler-2d}
    D = \frac{v}{\rho}\cos( \Theta - \alpha).
\end{equation}

\begin{theorem}\label{theorem1}
Assume that $\Theta$ and $\alpha$ are independent, such that $\Theta$ is uniform on $[0, 2\pi)$ and $\alpha$ has a general density $f_{\alpha}$. Then, in a 2D scenario, the Doppler angle $\Gamma$ and the Doppler shift are distributed as 
    
%
%
%
\begin{equation}\label{2d-Doppler}
\begin{aligned}
    \Gamma &{=} (\Theta-\alpha) \bmod 2\pi \stackrel{d}{\sim} \textnormal{U}[0,2\pi),
    \\
    f_{D}(x) &{=} \frac{1}{\pi} \left(\left(v/{\rho}\right)^2  {-} x^2 \right)^{-1/2}, \;\; \forall \ x \in \left[ -v/{\rho}, v/{\rho} \right]. 
\end{aligned}
\end{equation} 

\end{theorem}
%

%
%

\begin{proof}
    See Appendix A  
    for the proof.
\end{proof}
The PDF in \eqref{2d-Doppler}  corresponds to the Beta or Arcsine distribution, with parameter $\frac{v}{\rho}$, which is the magnitude of maximum Doppler shift.

Next, we discuss the statistics of the coherence time. Studying coherence time is crucial as it maps the Doppler shift (in the frequency domain) to the physical layer and to the time domain.  
%
%
%
The PDF of the coherence time  $T_c$, in a 2D setting, is provided below and the proof can be found in Appendix B. 
%
\begin{lemma}
The probability density of $T_c$ is given by
\begin{align}
    f_{T_c}(y)  {=} \frac{2}{\pi y} \!\!\left(\left(v y/\rho \mu \right)^2 \!\!{-} 1\right)^{-\half}, \forall \ y \in \left[\mu \rho/v, \infty \right).
    \label{eq:Tc-pdf}
\end{align}
\label{lemma:tc-pdf}
\end{lemma}
It is important to note, from the relation $T_c = \frac{\mu}{|D|}$ that, when the Doppler shift approaches zero, the coherence time tends to infinity. We will discuss the consequence of infinite coherence time on system design in the next section.


\subsection{3D statistical study} 

\begin{figure}
    \centering
    \includegraphics[width=0.45\textwidth]{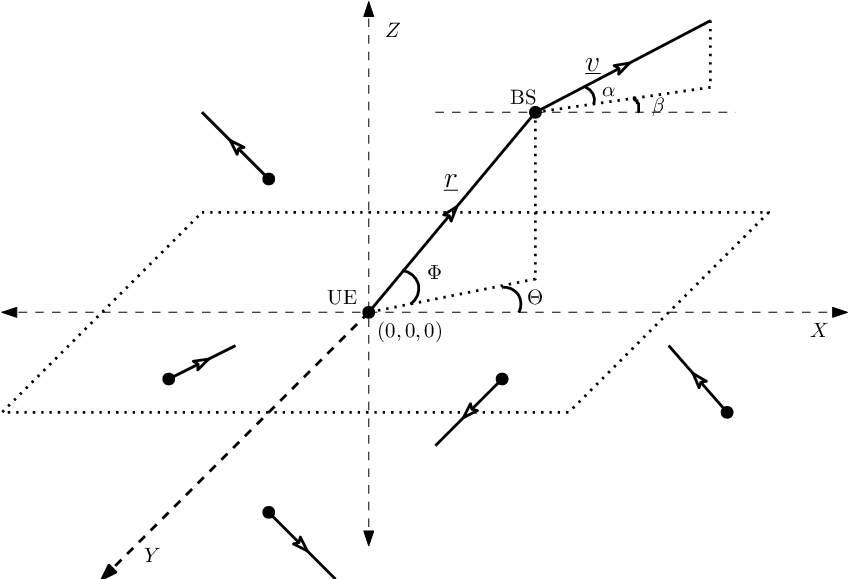}
    \caption{3D geometry of the vehicular network consisting of a  static UE receiver and mobile BS. }
    \label{fig:2}
\end{figure}

Consider a similar system model in 3D, with the BS locations modeled as homogeneous PPP $\Pcal:=\sum_{i\in \N} \delta_{X_i}$ on $\R^3$. 
%
%
As shown in Fig. \ref{fig:2}, for a BS located at a point $X$ of $\Pcal$, the LoS vector \underbar{$r$} is parameterized by an elevation angle $\Phi$ and an azimuth angle $\Theta$. 
Similarly, the DoM can be parameterized with an elevation angle $\alpha $ and azimuth $\beta$.  This 3D parameterization of the LoS and DoM vectors is represented in spherical coordinates as, $\underbar{$r$} {=} r[\cos \Phi  \cos\Theta\; \cos \Phi  \sin \Theta\; \sin \Phi]^T$ and $\underbar{$v$} {=} [\cos \alpha  \cos \beta\; \cos \alpha  \sin \beta\; \sin \alpha]^T$. The angles are measured from the $x-y$  plane and are considered independent. 
The cosine of the Doppler angle $\Gamma$ is then computed using inner product as 
\begin{equation}\label{3d-gamma}
\begin{aligned}
\cos\Gamma &{=} \frac{\langle\underbar{$r$}, \ \underbar{$v$}\rangle}{r v}
{=}  \cos \Phi \cos \alpha \cos(\Theta {-} \beta) +  \sin \Phi \sin \alpha. \\
\end{aligned}
\end{equation}

It may be noted that $\Phi \ \text{and} \ \alpha \in [-\frac{\pi}{2}, \frac{\pi}{2}]$ are polar (elevation) angles measured from the $x{-}y
$ plane whereas $\Theta \ \text{and} \ \beta \in [0, 2\pi)$ are the azimuth angles. 
The corresponding Doppler shift experienced at the static UE is 
    \begin{equation}
         D {=} \frac{v}{\rho} \cos \Gamma 
         {=} \frac{v}{\rho} (\cos \Phi \cos \alpha \cos(\Theta {-} \beta) {+}  \sin \Phi \sin \alpha).\nn
    \end{equation}

\begin{figure}
    \centering
    \includegraphics[width=0.4\textwidth]{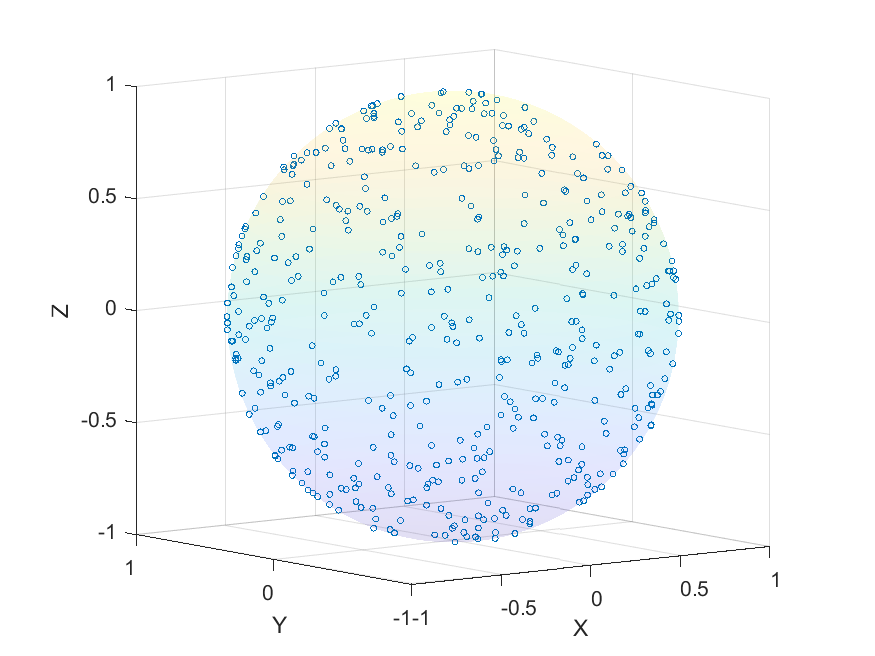}
    \caption{Illustration of 3D sampling of BS locations  for a given LoS norm $r$. Note that  for a given spherical radius $r$, the elevation angle $\Phi \stackrel{d}{\sim} \frac{1}{2} \cos \phi$ and the azimuth angle $\Theta \stackrel{d}{\sim} U[0, 2\pi)$ ensures the uniformity of BS locations in 3D space, in addition to Poissonianity.}
    \label{fig:3d-points}
\end{figure}

In principle, the Doppler angle $\Gamma\equiv \Gamma(\Theta,\Phi,\alpha,\beta)$ is a function of four angles $\Theta, \Phi, \alpha,$ and $\beta$.
It is notable that, in a 3D scenario, the distribution of the independent azimuthal angles $\Theta \ \text{and} \ \beta$ can be assumed to be uniform, i.e., $\Theta, \beta \stackrel{d}{\sim} \mbox{U}[0, 2\pi)$. In contrast, the distributions  of the polar angles $\Phi \ \text{and} \ \alpha$ are non-uniform in a spherical 3D setting, in order to ensure that for a given spherical radius the locations of BSs are uniformly distributed on the 3D space~\cite{Keeler} and \cite{pitman2012archimedes} (see Fig. \ref{fig:3d-points}). The distribution is provided below. 
\begin{lemma}\label{lemma-1}
The PDF of the elevation angles $\Phi, \alpha$ are 
\begin{align}
f_{\Phi}(\phi)=\frac{1}{2}\cos \phi \mbox{ and } f_{\alpha}(a)=\frac{1}{2}\cos a,
\label{phi-alpga}
\end{align}
for $\phi, a\in [-\pi/2,\pi/2]$, respectively.
\end{lemma}
\begin{proof}
      See Appendix C. 
\end{proof}

%

%
\begin{theorem}\label{theorem-3}
    In a generalized 3D scenario, assuming independent angles, the radial component of the velocity vector with respect to the Doppler angle $\Gamma$, and the corresponding Doppler shift are distributed as
    \begin{equation}
        \cos\Gamma\stackrel{d}{\sim}\textnormal{U}[-1, 1] \mbox{ and }  D\stackrel{d}{\sim} \textnormal{U}\left[-v/{\rho},v/{\rho} \right].\nn
    \end{equation}
\end{theorem}

\begin{proof}
    See Appendix D.  
\end{proof}

\begin{lemma}\label{coh-time-3d}
The probability density corresponding to the coherence time is
    \begin{equation}
    f_{T_c}(y) = \frac{\mu \rho}{v}  y^{-2}, \;\;\; \forall \,  y \in \left[\frac{\mu \rho}{v}, \infty \right).\nn
\end{equation}
\end{lemma}
\begin{proof}
    See Appendix E. 
\end{proof}
In the upcoming section, we propose two Doppler aware utility functions, which will be further used as key  metrics to evaluate system performance. 
%




\section{Doppler Aware UE Association policies}

The main objective of this work is to propose a Doppler-Shannon framework wherein the UE-BS association is based on the Doppler angle, Euclidean distance between UE and BS, and the fading gain of the UE-BS link, i.e., $f(\Gamma, R, K)$. This is a refinement w.r.t. the classical Shannon association where the association is a function of the closest distance and fading gain, i.e., $f(R, K)$.   

In this section, we propose two  different ``utilities'' capturing, (a)  the data reception at the UE and (b) the data rate achievable over a time selective fading channel. 
We then propose novel UE association policies based on these utilities. 

\subsection{Data received} 

To begin with, the Shannon rate $S = B \log_2 (1 + K R^{-\sigma})$, represents the achievable data rate  [bits/sec] from a BS located at distance $R$ to a UE. Here, $K$ denotes the channel fading gain. To capture only the amount of data being transmitted from a vehicular BS located at $X_i$ to a typical UE (assumed at the origin), we define the utility 
\begin{equation}
U_{i}(\Gamma_i,R_i,K) \triangleq T_{c}(\Gamma_i) B \log_2 (1 + K_i R_i^{-\sigma}) \ [\mbox{bits}],
\label{eq:U0}
\end{equation} 
where $R_i=|X_i|$, $K$ is the fading factor and $T_c(\Gamma_i)$ is coherence time corresponding to the Doppler angle $\Gamma_i$. Physically, the utility in (\ref{eq:U0}) denotes the amount of data (in bits) which the vehicular BS at $X_i$ can transmit to the typical UE  in a time--frequency resource block. 

To visualize the effect of the Doppler shift, we construct the corresponding BS  tessellation on $\R^d$ based on the utility defined in (\ref{eq:U0}). For a UE located at $z \in \R^d$, receiving utility $U_i(z):=U(\Gamma_i,|X_i-z|,K_i)$ from a BS located at $X_i\in \Pcal$. 
We consider the coverage region of a BS  
\begin{align}
\Rcal_i:=\{z\in \R^d: U_i(z)> U_j(z)\mbox{ for all }j\neq i\}.\nn
\end{align}

It is evident from \eqref{2d-Doppler}   that the minimum Doppler shift is experienced when the UE LoS is  perpendicular to the DoM, i.e., $D(\Gamma_i)$ close to $0$, when $\Gamma_i$ close to $\frac{\pi}{2}$. Then the corresponding  coherence time and value of the utility tends to infinity, i.e., $T_c(\Gamma_i) \rightarrow \infty$ (evident from \eqref{eq:Tc-pdf}) and $U_i \rightarrow \infty$. This phenomenon results in infinitely  elongated tessellations perpendicular to the DoM as $\Gamma_i$ is close to $\frac{\pi}{2}$ (depicted in Fig~\ref{fig:4}~(a) (2D) and Fig.~\ref{fig:9} (3D)). 

It is crucial to note that we perform our analysis in a time snap-shot (i.e., over a small time interval $\delta t)$. Hence, the position of the vehicular BSs (modeled as a PPP), and their respective tessellations change randomly with the time interval. The presence of unbounded infinitely long ``rays'' occurring in the spatial tessellations   is challenging as the evolution of coverage tessellations with time or UE handovers in a space-time dynamical system require spatial boundedness. 

We limit the extent of elongation of these ``rays'' by modifying the utility function through a \textit{limiting factor} $\zeta > 0$, as re-defined below. 

\begin{utility}[Data received]\label{utility1-proposition} 
The first utility function $U_1$, to model the amount of data being received in a coherence time interval (in bits),  at a typical UE from a mobile BS $X_i$ is  
\begin{equation}\label{utility1}
    U_{1,i}\triangleq \min\{T_{c}(\Gamma_i), \zeta \} B \log_2(1 + K_i R_i^{-\sigma}),
\end{equation}
where $U_{1,i} \equiv U_{1}(\Gamma_i,R_i,K_i)$, $\zeta$ denotes the limiting factor, $\Gamma$ the Doppler angle, $K$ the channel gain, $R_i=|X_i|$, and $\sigma$ the path-loss exponent (PLE). 
\end{utility}

The limiting factor $\zeta$ ensures that, mathematically, the utility does not tend to infinity, i.e., $\zeta < \infty  \implies U_i \neq \infty$. When $\zeta$ is very small, it mathematically nullifies the effect of the Doppler shift (through coherence time) on the utility, resulting in the  classical Voronoi tessellations.  Physically, the limiting factor $\zeta$ ensures that the UE do not associate with a far away mobile BS having the Doppler angle $\Gamma \ \text{close to} \ \pi/2$. Below we define the coverage tessellations arising from Utility 1. 

\textbf{Tessellation 1}:
For any $z\in \R^d$, we have  $U_{1,i}(z)$ denoting the data reception at $z$ from $X_i\in \Pcal$. The association rule corresponding to two vehicular BSs located at $X_i, X_j$ and a UE at location $z$, is according to $U_{1,i}(z) \lesseqqgtr U_{1,j}(z)$. This implies that
\begin{align}
\frac{\min\{T_{c}(\Gamma_{i}), \zeta\}}{\min\{T_{c}(\Gamma_{j}), \zeta\}} & \lesseqqgtr \frac{ \log_2 \left( 1+ K_{j} |X_j-z|^{-\sigma}\right)}{ \log_2 \left( 1+ K_{i} |X_i-z|^{-\sigma} \right)}.\nn
\end{align}
Hence, the comparison of the ratio of the truncated coherence times and therein the Doppler angle and the Shannon rate influence the coverage regions or the tessellation structure of the vehicular BS. The equality $U_{1,i}(z)  = U_{1,j}(z)$, along with $U_{1,i}(z), U_{1,j}(z) > U_{1,k}(z) $ for all $k \neq \{i, j\}$ where $X_k \in \mathcal{P}$, corresponds to a boundary point, if it exists, between the coverage cells of the BSs at $X_i, X_j$, of the tessellation.
The UE located at the boundary performs handover from one BS to the other.

\textcolor{magenta}{
}

\subsection{Data rate}

\begin{utility}[Data rate]\label{utility3-proposition}
    The second Utility $U_2$, captures the effect of channel estimation time on the amount of data rate (in bits per sec.)  obtained at a typical UE from a BS at $X_i$. For the channel estimation time, say $\varepsilon$, the utility is defined as 
\begin{equation}\label{utility3}
\hspace{-0.08in} U_{2,i} {\triangleq} \left( 1 {-} \frac{\varepsilon}{T_{c}(\Gamma_i)} \right)^{+} \!\!\! B \log_2 (1 {+} K_i R_i^{-\sigma}),
\end{equation}
where $U_{2,i} {\equiv }  U_{2}(\Gamma_i,R_i,K_i)$ and $|X_i|=R_i$. 
\end{utility}

 The second utility may be visualized as the data rate obtained in a coherence time interval after the allocation of $\varepsilon$ time units for the channel estimation, as is the norm in time division multi access (TDMA) systems. It may be noted that many upcoming V2X technologies including IoT systems are based on TDMA systems.
Further, the channel estimation time cannot exceed the coherence time interval, i.e. $\varepsilon \not> T_{c}$, as the utility becomes negative otherwise. 
It may be noted from \eqref{utility3} that a very small channel estimation time, $\varepsilon \ll T_{c}$, results in the utility resembling the classical Shannon data rate, thereby negating the effects of the Doppler (through the coherence time) on the system performance. Hence, it is noteworthy that efficient channel estimation is crucial for the success of wireless systems.  
Below we summarize the key observations: 
\begin{equation}
\begin{cases}
    \text{if} \ \varepsilon \ll T_{c} & \text{the effect of Doppler vanishes,} \\
    \text{if} \ \varepsilon < T_{c} & \text{the effect of Doppler is prominent}.
\end{cases}\nn
\end{equation}
\textbf{Tessellation 2}: Corresponding to Utility ~\ref{utility3-proposition}, 
for two vehicular BSs located at $X_i, X_j$ and a UE at location $z \in \R^d$, the BS association rule 
is  $U_{2,i}(z) \lesseqqgtr U_{2,j}(z)$. Similarly to the first utility, here also we describe the boundary between two neighboring coverage cells of the BSs located at $X_i, X_j$ with the helpf of the equality $U_{2,i}(z) = U_{2,j}(z)$, together with the inequality $U_{2,i}(z), U_{2,j}(z) > U_{2,k}(z)$ for all $k \neq \{i, j\}$, where $X_k \in \mathcal{P}$. See Figs.~\ref{fig:5}~(a)--(c) and Fig.~\ref{fig:11} for the tessellations.

It can be inferred that, for very small $\varepsilon \ll T_{c}$, the UE association is primarily happening by the classical Shannon association. As $\varepsilon < T_{c}$, the effect of Doppler on the UE association becomes more and more prominent.

\begin{remark}
Note that, $T_{c} {>} \varepsilon$ implies
$\frac{\mu \rho}{v |\cos \Gamma|} {>} \varepsilon$, which in turn implies that 
$\frac{\mu c}{v |\cos \Gamma| f_c} {>} \varepsilon$,
where $f_c$ is the frequency center of the communication band and $c$ is speed of light. Hence, we infer that communication in higher frequency bands will result in reducing the coherence time interval, resulting in limiting the network performance. 
Similarly, a higher magnitude of velocity $v$ will result in reducing the coherence time interval, thereby affecting network performance. 
\end{remark}

In the upcoming sections, we analytically evaluate the performance of the new Doppler-Shannon based association policies and compare them with the classical Shannon based framework. 


\section{Doppler-Shannon based Performance Evaluation in 2D}

In this section, we evaluate the performance of the proposed Doppler-Shannon association frameworks using stochastic geometry tools, in a 2D setting. The performance evaluation will be further extended to 3D  in the next section.

\subsection{ Doppler-Shannon association framework}

The objective is to associate the UE with the BS from which the UE receives maximum utility. The location of the associating BS is
\begin{equation}
    x^* = \arg \max_i U_i.\nn
\end{equation}
Here, $U_i=\{ U_{1,i}, U_{2,i}\}$ denotes a general utility  as a function of the BS location w.r.t. a typical UE. 
For simplicity, we denote the utility function as $U$ in general for the rest of the article and specify it  wherever necessary.   

Before outlining the performance evaluation in depth,  we prove the rotational invariance of the Doppler shift in a 2D setting, a property which we term as {\em isotropy}. Meaning that its statistical properties do not change with axis rotations.  
\begin{lemma}\label{RI2D}
The Doppler angle $\Gamma$ is rotationally invariant.
\end{lemma}
\begin{proof}
The initial LoS and DoM vectors are given by $\underbar{$r$} {=}[r\cos\Theta\; r\sin\Theta]^T$ and 
 $\underbar{$v$} {=} [v\cos \alpha\; v \sin \alpha]^T$, in terms of the angles $\Theta$ and $\alpha$, respectively. The Doppler angle $\Gamma$ is given by the inner product $\frac{ \langle\underbar{$r$}, \ \underbar{$v$}\rangle }{rv}$. Since the inner product preserves angle between two vector, the Doppler angle remains invariant with respect to any rotation. 
%
%
\end{proof}

From Theorem~\ref{theorem1} and Lemma~\ref{RI2D}, we deduce the following corollary about the isotropic behavior of the Doppler statistics. 
\begin{corollary}\label{cor:cor1}
The DoM has no effect on the distribution of the radial component of velocity along the LoS, i.e.,
\begin{equation}
\Gamma(\Theta, \alpha) = \cos(\Theta {-} \alpha) \stackrel{d}{\sim} \cos \Theta \equiv \Gamma(\Theta).\nn
\end{equation}
%
 
\end{corollary}

\begin{remark}
In view of Corollary~\ref{cor:cor1}, since the Doppler angle and further  the Doppler shift and coherence time are oblivious of the DoM, in a 2D setting,  we define the utility purely as a function of the angle of the LoS, i.e., $\Theta$, along with the Euclidean distance  ($R$) and the channel fading gain ($K$). Hence, i.e., $U_i  \equiv  U(\Gamma_i, R_i, K_i) \equiv U(\Theta_i, R_i, K_i)$. Equivalently, $T_c(\Gamma) \equiv T_c(\Theta)$.
%
In other words, the utility function is  defined in terms of the marked PPP $\tilde \Pcal:=\sum_i\delta_{(R_i,\Theta_i),K_i}$, where the Doppler angle is implicitly defined in the polar coordinate of the spatial location of the BSs. 
\end{remark}

Below we present our main result regarding the outage probability at a UE to meet a QoS $\tau$.  
\begin{theorem}\label{theorem-6}
For mobile nodes modeled spatially as a 2D Poisson point process the outage probability at a typical UE, to meet the user QoS $\tau$ 
, is given by
\begin{equation}
   \hspace{-0.1in} \mathbb{P}\left(\max_i U(\Theta_i, R_i, K_i) {<} \tau \right) {=}  e^{- \lambda \pi   \mathbb{E}_{K, \Theta} \left[ K^{\frac{2}{\sigma}}g(\Theta, \tau)^2 \right]},
   \label{eq:OPU}
\end{equation}
where $K^\frac{1}{\sigma} g(\Theta,\tau)$ represents the cutoff distance, as a function of the LoS angle and the QoS,  signifying the distance from the typical UE to the optimal BS. \label{theorem:Cov-S}
\end{theorem}

\begin{proof}
The outage probability w.r.t. the user QoS $\tau$ is
\begin{align}
\hspace{-0.1in}\mathbb{P}\left(\max_i U(\Theta_i, R_i, K_i) {<} \tau \right) &{=} \mathbb{E}\left[\prod_i \mathbbm{1}_{\{U(\Theta_i, R_i, K_i) < \tau\}}  \right].
\label{eq:OPU1}
\end{align}
Using the PGFl of max-shot noise in a Poisson setting in (\ref{eq:OPU1}) we have
\begin{align}   
\lefteqn{\mathbb{E}\left[\prod_i \mathbbm{1}_{\{U(\Theta_i, R_i, K_i) < \tau\}}  \right]}\nn\\
&=\mathbb{E}_{\mathcal P}\left[\prod_i \E\left[\mathbbm{1}_{\{U(\Theta_i, R_i, K_i) < \tau\}}\right]\vert \mathcal P  \right]\nn\\
&= \exp \left( - \lambda \int_0^{2\pi}\!\!\! \int_0^\infty\!\! \left(1{-} \E_K\left[\mathbb{P}(U(\theta, r, K) < \tau)\right]\right) r \dd r   \dd \theta \right) \nn  \\
%
%
&= \exp \left( - \lambda \int_0^{2\pi} \!\!\int_0^\infty \!\!\E_K\left[\mathbb{P}(U(\theta, r, K) \geq \tau) \right]  r \dd r   \dd \theta \right).
\label{eq:r-eqv}
\end{align}
To solve further, we claim the following equivalence of events hold for all utility functions (for a given $\theta, \tau$ and $K$)
\begin{equation}\label{eq:g-eqv}
\left\{U(\theta, r, K) \geq \tau\right\} \equiv \left\{r \leq  K^\frac{1}{\sigma}g(\theta, \tau) \right\},
\end{equation}
using a function $g(\theta,\tau)$. Here, the expression  $K^\frac{1}{\sigma} g(\theta, \tau)$ denotes a cut-off distance, as a function of the LoS angle and the QoS, 
and is different for different utility function. 
To demonstrate the above equivalence and existence of such a function $g(\theta,\tau)$, let us consider the Utility~\ref{utility1-proposition},
\begin{equation}
U =   \min\{T_{c}(\theta), \zeta \} B \log_2 (1 {+} K r^{-\sigma}),
\label{eq:U1} 
\end{equation}
where $T_c(\theta) = \mu \slash |D(\theta)|$. As we algebraically solve the inequality, we have a string of equivalence of the events from (\ref{eq:g-eqv}),  
\begin{align}
\{U(\theta, r, K) \geq \tau\}
&{=} \left\{\min\{T_{c}(\theta), \zeta \} B \log_2 (1 {+} K r^{-\sigma})  \geq \tau\right\} \notag \\
&{=} \left\{r^{-\sigma} \geq (2^{\tau/\min\{T_{c}(\theta), \zeta \} B} - 1)/K\right\}  \nn\\
&\equiv \left\{ r \leq K^\frac{1}{\sigma} g(\theta, \tau)\right\}. \nn
\end{align}
%
For the first utility (\ref{eq:U1}) we have
\begin{align}    g(\theta, \tau) &:= \left(2^{\tau/\min\{T_{c}(\theta), \zeta \} B} - 1\right)^{-\frac{1}{\sigma}}.
\label{eq:g-cutoff-U1}
\end{align}
The cut-off function for other utility functions are derived in Remark~\ref{remark:cutoff-g}. Using the equivalence (\ref{eq:g-eqv}), the last term in (\ref{eq:r-eqv}) becomes
\begin{align}
\lefteqn{\hspace{-0.2in}\exp \left( - \lambda \int_{0}^{2 \pi} \int_{0}^{\infty} \E_K\left[\mathbb{P}\left( r \leq K^\frac{1}{\sigma}g(\theta, \tau) \right)\right]  r \dd r   \dd \theta \right)}\nn\\
&= \exp \left( - \lambda  \int_{0}^{2 \pi}  \E_K\left[ \int_{0}^{K^\frac{1}{\sigma} g(\theta, \tau)}  r \dd r \right]\dd \theta \right)\nn \\
&= \exp \left( - \frac{\lambda}{2}  \int_{0}^{2 \pi} \E_K\left[ K^\frac{2}{\sigma}  \right] g(\theta, \tau)^2  \  \dd \theta \right) \nn \\
%
%
&= \exp \left( -\lambda  \pi \E_K\left[ K^\frac{2}{\sigma}  \right] \mathbb{E}_{\Theta} \left[ g(\Theta, \tau)^2 \right]   \right), \tageq \label{2d-DS-stochastic-geo-result}
\end{align}
since the fading $K$ is independent of the angle $\Theta$. This proves the result.
Moreover, the corresponding success or coverage  probability is given as 
\begin{equation}
\resizebox{\hsize}{!}{$\mathbb{P}\left(\max_i U(\Theta_i, R_i, K_i) {\geq} \tau \right) 
{=} 1 {-} e^{ - \lambda  \pi   \mathbb{E}_{K} \left[K^{\frac{2}{\sigma}}\right]\mathbb{E}_{\Theta} \left[g(\Theta, \tau)^2\right]}.$}\nn
\end{equation}
\end{proof}
\begin{remark}
Note that the probability that none of the BSs covers the typical user is computed as the void probability of the marked PPP $\tilde \Pcal:=\sum_i\delta_{(R_i,\Theta_i),K_i}$. In other words, it is the probability that there does not exist any BS that covers the UE, with the ensemble average over all possible LoS angle $\Theta$ and the fading $K$.
\label{remark:ea}
\end{remark}

In the remarks below, we derive the expression for $g(\Theta, \tau)$ and $\mathbb{E}_{K} \left[K^\frac{2}{\sigma}\right]$. 
\begin{remark}
With the coherence time simplified as a function of the LoS angle, $T_c(\theta) = \mu \slash |D (\theta)|$, the cut-off distance function $g(\theta, \tau)$ for the first utility $U_1$ is derived in \eqref{eq:g-cutoff-U1}.  
The cut-off distance function $g(\theta, \tau)$ for utility $U_2$ is
\begin{equation}
g(\theta, \tau)  := \left(2^{\tau/\left( 1 - \frac{\varepsilon}{T_{c}(\theta)} \right)^{+} B} - 1\right)^{-\frac{1}{\sigma}}.
\label{eq:g-cutoff-U3}
\end{equation}
%
\label{remark:cutoff-g}
\end{remark}
%
%
\begin{remark}
In general, the expectation of the cut-off distance function is given as
\begin{equation*}
    \mathbb{E}_{\Theta} \left[g(\Theta, \tau)^2\right] = \int_{0}^{2\pi} \frac{1}{2 \pi} g(\theta, \tau)^2 \dd \theta.
\end{equation*}
We substitute the cut-off distances provided in \eqref{eq:g-cutoff-U1} 
and \eqref{eq:g-cutoff-U3}, corresponding to $U_1, U_2$, 
respectively. For $T_c(\theta) = \mu \slash |D (\theta)|$,  we get
\begin{align*}
    \text{For} \ U_1: \mathbb{E}_{\Theta}\!\left[g(\Theta, \tau)^2\right] &{=} \!\!\int_{0}^{2\pi}\!\!\!\! \frac{1}{2 \pi} \left(2^{\tau/\min\{T_{c}(\theta), \zeta \} B} {-} 1\right)^{-\frac{2}{\sigma}} \!\!\dd \theta, 
    \\
    \text{For} \ U_2: \mathbb{E}_{\Theta}\!\left[g(\Theta, \tau)^2\right] &{=}\!\! \int_{0}^{2\pi}\!\!\!\! \frac{1}{2 \pi} \left(2^{\tau/\left( 1 - \frac{\varepsilon}{T_{c}(\theta)} \right)^{+} B} {-} 1\right)^{-\frac{2}{\sigma}} \!\dd \theta. 
\end{align*}
The integral expressions cannot be further decomposed into closed form expressions and require to  numerically simulated.

The expectation of the fading term $\mathbb{E}_K[K^\frac{2}{\sigma}]$, for a general Nakagami fading environment with shape parameter $m$ and  scale parameter $\Omega$ is 
\[
\mathbb{E}_K[K^{\frac{2}{\sigma}}] {=} \frac{\Omega^\frac{2}{\sigma}}{\Gamma(m)} \Gamma\left( m {+} \frac{2}{\sigma} \right).
\]
For a Rayleigh fading environment with parameters $m=1$ and $\Omega=1$, the expression simplifies to $\mathbb{E}_K[K^\frac{2}{\sigma}] {=} \Gamma\left( 1 {+} \frac{2}{\sigma} \right)$. The proof of these closed form expressions are shown in Appendix 
F.
\end{remark}
\begin{remark}\label{exp-utility}
The expected value of max utility is
\begin{align}
\mathbb{E}[\max_i U(\Theta_i, R_i, K_i)]
&{=} \int_{0}^{\infty}\!\!\!\! \mathbb{P}\left(\max_i U(\Theta_i, R_i, K_i) {\geq} \tau \right)\, \dd \tau\nn\\
& \hspace{-0.7in}=\int_{0}^{\infty}\!\! \left( 1 {-} e^{- \lambda  \pi   \mathbb{E}_{K} \left[K^\frac{2}{\sigma}\right]\mathbb{E}_{\Theta} \left[g(\Theta, \tau)^2\right]} \right)\dd \tau.\nn
\end{align}
\end{remark}
Note that for simplicity of the analysis, we consider a generalized $\tau$,  corresponding to the utility $U$. As shown in Remark \ref{exp-utility}, we have computed the expected utility by averaging over all possible values of $\tau$. Hence, the results derived in the section are general in nature, and are independent of the utility and its corresponding QoS.

In the upcoming subsection, we compare the proposed Doppler-Shannon framework with the closest distance association based Shannon framework. Note that, intuitively, the coverage probability of associating with the BS providing the ``maximum utility'' (which was the Doppler-Shannon case) will always be greater than (or equal to) the coverage probability of  associating with the nearest BS. In order to quantify the extent of gains observed, we manage to compare the coverage probabilities analytically using Jensen's inequality, as shown in the next subsection.

\subsection{Comparison with classical Shannon policy}

\begin{theorem}
The coverage probability to meet user QoS  through the proposed  Doppler-Shannon based association rule is always greater than or equal to than  that  obtained through the classical Shannon based association rule. 
\label{theorem:CovP-S-DS}
\end{theorem}
\begin{proof}
In the classical Shannon based association rule the distribution of the distance, say $R_o$, to the closest BS is provided by a Rayleigh density
\[
f_{R_o}(r):= 2\lambda\pi r e^{-\lambda\pi r^2}, 
\]
for $r\geq 0$, which can be determined using void probability of PPP $\Pcal$, as derived in \cite{andrews2016primer}.
Then, given the relative angle of motion, say $\Theta_o$ and the fading, say $K_o$ for the nearest station, the coverage probability is 
\begin{align}
\lefteqn{\mathbb{P}[U(\Theta_o, R_o, K_o) > \tau | \Theta_o, K_o]}\nn\\
&=  \int_{0}^{\infty} 2\lambda  \pi r e^{-\lambda \pi r^2} \one{\left\{ r \leq K_o^\frac{1}{\sigma} g(\Theta_o, \tau)\right\}}  \dd r\nn \\
& = 2\lambda  \pi \int_{0}^{K_o^\frac{1}{\sigma}g(\Theta_o,\tau)} e^{-\la \pi r^2} r \ \dd r.
\label{eq:shannon1}
\end{align}
By the change of variable  $\lambda \pi r^2 = u$,  the last term in (\ref{eq:shannon1}) is
\begin{align}
\int_{0}^{\lambda \pi K_o^\frac{2}{\sigma}g(\Theta_o,\tau)^2} e^{-u} \dd u = 1 - e^{-\lambda \pi K_o^\frac{2}{\sigma} g(\Theta_o, \tau)^2}.\label{eq:s2}
\end{align}
Hence, from (\ref{eq:shannon1}) and (\ref{eq:s2}) the outage probability is
\begin{align}
%
\mathbb{P}[U(\Theta_o, R_o, K_o) {\leq} \tau] &{=} \E_{\Theta_o, K_o}\!\!\left[\mathbb{P}(U(\Theta_o, R_o, K_o) {\leq} \tau | \Theta_o, K_o)\right]\nn\\
&{=}\mathbb{E}_{\Theta_o, K_o} \left[e^{-\lambda \pi K_o^\frac{2}{\sigma}g(\Theta_o, \tau)^2} \right] \tageq \label{2d-shannon-analytic}.
\end{align}
Previously, for our Doppler-Shannon based association rule, the outage probability from \eqref{2d-DS-stochastic-geo-result} is
\begin{equation}
\mathbb{P}\left(\max_i U(\Theta_i, R_i, K_i) {<} \tau \right) 
{=} e^{- \lambda \pi  \mathbb{E}_{\Theta, K} \left[K^\frac{2}{\sigma}g(\Theta, \tau)^2 \right]}.
\label{eq:s3} 
\end{equation}
By Jensen's inequality, we have from (\ref{eq:s3}) and \eqref{2d-shannon-analytic} that
\begin{align}
     e^{- \lambda \pi  \mathbb{E}_{\Theta, K} \left[ K^\frac{2}{\sigma} g(\Theta, \tau)^2 \right]} &\leq \mathbb{E}_{\Theta, K}\left[e^{-\lambda \pi K^\frac{2}{\sigma} g(\Theta, \tau)^2}\right].\nn
\end{align}
In other words, comparing the coverage probabilities we have
\begin{align}
\hspace{-0.07in}\mathbb{P}\left(\max_i U(\Theta_i, R_i, K_i)\! >\!  \tau \right) &\geq \mathbb{P}\left(U(\Theta_o, K_o, R_o) \!>\! \tau\right).\nn
\end{align}
Hence, we have proved analytically that the coverage probability of the proposed Doppler-Shannon association framework will always beat the coverage probability in the classical Shannon based association framework. This demonstrates that the Shannon association is not optimal in vehicular networks, and shows the need for a Doppler aware association strategy.   
\end{proof}
%

%

%

\section{Doppler-Shannon based Performance Evaluation in 3D}
\label{secV}



In this section, we extend the analysis of the Doppler-Shannon based UE association framework, to a 3D PPP setting. As in the 2D case, in the 3D setting also, we derive the outage probabilities for all the proposed utility functions using stochastic geometry tools.

\subsection{Doppler-Shannon association in 3D}
%

Before, outlining the performance analysis, we first prove the rotational invariance of the Doppler shift in a 3D PPP setting. 

\begin{lemma}\label{RI3D}
In the 3D PPP setting, the Doppler angle $\Gamma$ is rotationally invariant. 
\end{lemma}
\begin{proof}
    See Appendix 
    G.
\end{proof}

Below, we show that the Doppler statistics are identical to Theorem \ref{theorem-3}  in an isotropic setting. We consider an isotropic setting, wherein the elevation and azimuth of the DoM are known to the UE, hence, randomness is only due to the LoS angles. (The DoM angles may be assumed zero, i.e., $\alpha {=} \beta {=} 0$).

With this consideration, we denote the Doppler angle with respect to the LoS  as $\Gamma_1\equiv \Gamma_1(\Theta,\Phi)$. Then %
\begin{equation}\label{3d-gamma-case2}
    \cos \Gamma_1 =  \cos \Phi \cos \Theta. 
\end{equation} 
\begin{theorem}\label{theorem-5}
In an isotropic setting, the distribution of the cosine of the Doppler angle ($\Gamma_1$) and the corresponding Doppler shift are identical to those in Theorem \ref{theorem-3} and are given as
\begin{equation}
\begin{aligned}
\cos\Gamma_1\stackrel{d}{\sim}\textnormal{U}[-1, 1] \mbox{ and }  D\stackrel{d}{\sim} \textnormal{U}\left[-v/{\rho},v/{\rho} \right].\nn
\end{aligned}
\end{equation}
\end{theorem}

\begin{proof}
    See Appendix 
    H.
\end{proof}
The following corollary justifies that, in a 3D setting, the statistics of the Doppler angle and subsequently the Doppler shift and coherence time  (Theorem~\ref{theorem-5}), are identical  to the statistics obtained in  the general 3D case (Theorem~\ref{theorem-3}).  
\begin{corollary}
The Doppler angle $\Gamma$ is independent of the DoM angles $\alpha,\beta$, i.e.,
\begin{equation}\label{eq:corollary2}
\Gamma (\Theta, \Phi, \alpha, \beta) \stackrel{d}{\sim} \Gamma_1 (\Theta, \Phi).\nn
\end{equation}
\label{cor:cor2}
\end{corollary}
A confirmation of the analytical proof is performed through numerical simulation and is shown in Fig.~\ref{fig:3} (in Section~\ref{secVI}).

\begin{remark}
Note that in the 3D setting, the Doppler angle is a function of four independent angles, i.e., $\Gamma(\Theta, \Phi, \alpha, \beta)$. Accordingly, the utility functions, $U\equiv U(\Gamma,R,K)$ are also  functions of these four angles in addition to the Euclidean distance $R$ and the fading gain $K$. From Corollary~\ref{cor:cor2}, the utility function can be equivalently expressed statistically as, $U(\Gamma,R,K) \equiv U(\Gamma_1, R, K) \equiv U(\Theta, \Phi, R, K)$. Equivalently, $T_c(\Gamma) \equiv T_c (\Theta, \Phi)$.
%
\end{remark}

\begin{theorem}\label{theorem-8}
For mobile nodes modeled spatially as a 3D Poisson point process, the outage probability at a typical UE to meet the user QoS $\tau$ is given by
\begin{equation}\label{3d-outage} \mathbb{P}\left(\max_i U(\Theta_i, \Phi_i, R_i, K_i) {<} \tau \right) {=} e^{- \frac{4\lambda  \pi}{3} \mathbb{E}_{\Phi, \Theta, K}\left[K^{\frac{3}{\sigma}}g(\Theta, \Phi, \tau)^3\right]},\nn
\end{equation}
with $g(\Theta, \Phi, \tau)$ representing the cut-off distance function. 
\end{theorem}


\begin{proof}
As for the 2D setting, we use the PGFL of the max shot-noise in a 3D Poisson setting and obtain
\begin{equation*}
    \resizebox{\hsize}{!}{$\begin{aligned}
    &\mathbb{P}\left(\max_i U_i(\Theta_i, \Phi_i, R_i, K_i) < \tau \right) \\
    &{=} \exp \left( - \lambda \int_{ \mathbb{R}^{3}} \E_K\left[\mathbb{P}(U(\theta, \phi, r, K) {\geq} \tau) \right] \dd \underbar{x} \right) 
    \\
     &{=} \exp \left( - \lambda \int_{-\pi/2}^{\pi/2}  \int_0^{2\pi} \int_0^\infty \E_{K}\left[\mathbb{P}(U(\theta, \phi, r, K) {\geq} \tau)\right] r^2 \dd r \dd \theta \cos \phi \dd \phi \right).    
    \end{aligned}$}  
\end{equation*}
Using an equivalence similar to \eqref{eq:g-eqv}, there exists a function $g(\theta, \phi,\tau)$ such that
\[
\left\{U(\theta, \phi, r, K) {\geq} \tau\right\} \equiv \left\{r \leq K^{1/\sigma} g(\theta, \phi,\tau)\right\}.
\]
Solving the integral, we get the outage probability 
\begin{align}
\lefteqn{\mathbb{P}\left(\max_i U_i(\Theta_i, \Phi_i, R_i, K_i) < \tau \right)}\nn\\
&{=} \exp \left( - \lambda \int_{-\frac{\pi}{2}}^{\frac{\pi}{2}} \int_{0}^{2 \pi} \E_{K}\left[\int_{r=0}^{K^{1/\sigma}g(\theta, \phi, \tau)} r^2 \dd r \right]  \dd \theta \cos \phi \dd \phi\right)\nn\\
&{=} \exp \left( - \frac{\lambda}{3} \int_{-\frac{\pi}{2}}^{\frac{\pi}{2}} \int_{0}^{2 \pi} \E_{K}\left[K^{\frac{3}{\sigma}} g(\theta, \phi,\tau)^3\right]   \dd \theta \cos \phi \dd \phi  \right)\nn\\
&{=} \exp \left( - \frac{4\lambda \pi}{3} \E_{K}\left[K^{\frac{3}{\sigma}}\right] \mathbb{E}_{\Phi, \Theta}\left[g(\Theta, \Phi, \tau)^3\right]  \right). \tageq \label{3d-DS-stochastic-result}
\end{align}
This proves the result. 
Note that the corresponding success probability  is
\begin{align}
\resizebox{\hsize}{!}{$\mathbb{P}\left(\max_i U_i(\Theta_i, \Phi_i, R_i, K_i) {\geq} \tau \right){=} 1 {-} e^{- \frac{4\lambda  \pi}{3} \mathbb{E}_{K}\left[K^{\frac{3}{\sigma}}\right]\mathbb{E}_{\Phi, \Theta}\left[g(\Theta, \Phi, \tau)^3\right]}.$}\nn
\end{align}
\end{proof}

\begin{remark}\label{cutoff-funcs-3d}
With the coherence time simplified as, $T_c(\t, \phi) = \mu \slash |D(\theta, \phi)|$, the cut-off  distance functions $g(\theta, \phi,\tau)$ for the proposed utilities are determined as follows. For the first  utility $U_1$, the cut-off function is 
\begin{align}
g(\theta, \phi, \tau) &{=} \left(2^{\tau/\min\{T_{c}(\theta, \phi), \zeta \} B} - 1\right)^{-\frac{1}{\sigma}}.\nn
\end{align}
For the second utility $U_2$, the cutoff function is
\begin{align}
g(\theta,\phi,\tau) &= \left(2^{\tau/\left( 1 - \frac{\varepsilon}{T_{c}(\theta, \phi)} \right)^{+} B} - 1\right)^{-\frac{1}{\sigma}}.\nn
\end{align}
%
\end{remark}

\begin{remark}
In general, the expectation of the cut-off distance function is given as
\begin{equation*}
    \mathbb{E}_{\Phi, \Theta} \left[g(\Phi, \Theta, \tau)^3\right] {=} \int_{-\frac{\pi}{2}}^{\frac{\pi}{2}} \frac{1}{2} \cos \phi \int_{0}^{2\pi}  \frac{1}{2 \pi} g(\phi, \theta, \tau)^3   \dd \theta \dd \phi.
\end{equation*}
We substitute the cut-off distances provided in Remark \ref{cutoff-funcs-3d}, corresponding to the respective utility.

For a general Nakagami fading environment with shape parameter $m$ and  scale parameter $\Omega$, the expectation of the fading term $\mathbb{E}_K\left[K^{\frac{3}{\sigma}}\right]$ is
\begin{equation}
\mathbb{E}_K\left[K^{\frac{3}{\sigma}}\right] = \frac{\Omega^{\frac{3}{\sigma}}}{\Gamma(m)} \Gamma\left( m {+} \frac{3}{\sigma} \right).\nn
\end{equation}
For a Rayleigh fading environment with $\Omega=1$ and $m=1$, the expression simplifies to $\mathbb{E}_K[K^{\frac{3}{\sigma}}] = \Gamma\left( 1 + \frac{3}{\sigma} \right).$
The proof of these closed form expressions are shown in Appendix  
F.
\end{remark}

A generalized algorithm summarizing the proposed 2D and 3D Doppler-Shannon framework is provided in Algorithm \ref{Algo1}.

\begin{algorithm}[t]
\SetAlgoLined 
\footnotesize
\KwResult{Coverage probability $P_c(\tau):=\mathbb{P}(\max_i U_i \geq \tau)$}
\textbf{Input}: $\lambda, v, \rho, B, \mu, \zeta, m, \Omega, \sigma, \tau$ 
\\
\textbf{Initialization}: 
For a given $\lambda$, model BS locations as a PPP: $\mathcal{P} = \sum_i \delta_{X_i}$ \\
Location of BSs: $X_i \in \mathbb{R}^d, \ d =2, 3$  \\
For 2D and 3D: i.i.d $\Theta_i\stackrel{d}{\sim} \mbox{U}[0, 2\pi)$ \\
For 3D: $\Phi\stackrel{d}{\sim} \frac{1}{2}\cos \phi \in [-\frac{\pi}{2}, \frac{\pi}{2}]$ \\
\textbf{Sampling the angles}: Generate $N$ i.~i.~d. random variables, having density $\mbox{U}[0,1], \ \underbar{$x$} = (x_0,x_1,\cdots,x_{N-1})= \mbox{rand}(N,1)$ \\
$x_j\stackrel{d}{\sim} \mbox{U}[0,1],\;\; \forall j \in \{0, 1, \ldots N-1\}$ \\
Samples of $\Phi$: $\phi_j = \sin^{-1}{(2x_j - 1)}$ (inverse sampling) \\
Samples of $\Theta$:  $\theta_j = 2\pi \  \mbox{rand}(N,1)$ \\
\textbf{Methodology:} \\
$|D_i| = \frac{v}{\rho} |\cos\Phi_i \cos \Theta_i| $ \\
Compute $T_{c_i} = \frac{\mu}{|D_i|}$ \\
Compute $R_i = ||X_i||$ 
\\
Compute the utility of choice, $U_i = U_{1,i}, U_{2,i}$
\\
In 2D: $U_i = U(\Theta_i, R_i, K_i)$, 
In 3D: $U_i = U(\Theta_i, \Phi_i, R_i, K_i)$ 
\\
For $U_{1,i}:$ $g(\theta, \tau) := \left(2^{\tau/\min\{T_{c_i}, \zeta \} B} - 1\right)^{-\frac{1}{\sigma}}$ \\
For $U_{2,i}:$ $g(\theta, \tau)  := \left(2^{\tau/\left( 1 - \frac{\varepsilon}{T_{c_i}} \right)^{+} B} - 1\right)^{-\frac{1}{\sigma}}$ 
\\
For 2D: $\mathbb{E}_K\left[K^{\frac{2}{\sigma}}\right] = \frac{\Omega^{\frac{2}{\sigma}}}{\Gamma(m)} \Gamma\left( m {+} \frac{2}{\sigma} \right)$ \\
For 3D: $\mathbb{E}_K\left[K^{\frac{3}{\sigma}}\right] = \frac{\Omega^{\frac{3}{\sigma}}}{\Gamma(m)} \Gamma\left( m {+} \frac{3}{\sigma} \right)$ \\
\textbf{Compute the outage probability:} \\
For 2D: $P_{out}(\tau) = \exp\left(- \lambda  \pi   \mathbb{E}_{K} \left[K^\frac{2}{\sigma}\right]\mathbb{E}_{\Theta} \left[g(\Theta, \tau)^2\right]\right)$ \\
For 3D: $P_{out}(\tau) = \exp\left(- \frac{4\lambda  \pi}{3} \mathbb{E}_{K}\left[K^{\frac{3}{\sigma}}\right] \mathbb{E}_{\Phi, \Theta}\left[g(\Theta, \Phi, \tau)^3\right]\right)$
\\
\textbf{Coverage probability:} $P_c(\tau) = 1  - P_{out}(\tau)$
\caption{Coverage computation methodology}\label{Algo1}
\end{algorithm}

\subsection{Classical Shannon association in 3D}
Let $R_o$ be the distance to the nearest BS. Then $R_o$ follows a  density, for $r\geq 0$, given by (See Appendix 
I 
for the proof)
\begin{equation}
f_{R_o}(r) {=}  4 \pi \lambda r^2 e^{- \frac{4 \pi \lambda}{3}r^3}.\nn
\end{equation}
%
%

\begin{theorem}\label{theorem-9}
In a 3D PPP setting, the coverage probability in the proposed Doppler-Shannon association policy is always greater than or equal to the classical Shannon based association for any QoS $\tau$ unit, i.e.,  
    \begin{align} 
    \mathbb{P}\left(\max_i U(\Theta_i, \Phi_i, R_i, K_i) {\geq} \tau \right)  &{\geq} \mathbb{P}(U(\Theta_o, \Phi_o, R_o, K_o) {>} \tau).\nn
\end{align}
\end{theorem}


\begin{proof}
Given the 
azimuth and elevation angles $(\Theta_o, \Phi_o)$ and the fading factor $K_o$, the coverage probability from the nearest BS, at a distance say $R_o$, is  
\begin{align}
\lefteqn{\mathbb{P}\left(U(\Theta_o, \Phi_o, K_o, R_o) > \tau | \Theta_o, \Phi_o, K_o\right)} \nn \\
& =  \int_{0}^{\infty} 4 \pi \lambda r^2 e^{- \frac{4 \pi \lambda}{3} r^3} \one{\left\{r \leq K_o^\frac{1}{\sigma}g(\theta_o, \phi_o,\tau)\right\}}  \dd r\nn\\
&=  \int_{0}^{K_o^{1/\sigma}g(\theta_o, \phi_o,\tau)}  4 \pi \lambda r^2 e^{- \frac{4 \pi \lambda}{3}r^3}\dd r. \label{eq:s1-3d}
\end{align}
Substituting $\frac{4 \pi \lambda}{3} r^3 = u$, the integral is modified as 
\begin{align}
\int_{0}^{\frac{4 \pi\lambda}{3}  K_o^{\frac{3}{\sigma}}g(\theta_o, \phi_o, \tau)^3} \hspace{-0.3in}e^{-u} \dd u = 1 - e^{- \frac{4 \pi\lambda}{3} K_o^{\frac{3}{\sigma}}g(\theta_o,\phi_o,\tau)^3}. 
\label{eq:s2-3d}
\end{align}
Therefore from (\ref{eq:s1-3d}) and (\ref{eq:s2-3d}) the outage probability is
\begin{equation}
\mathbb{P}\left(U(\Theta_o, \Phi_o, K_o, r_o) {\leq} \tau\right) {=}  \mathbb{E}_{\Theta_o, \Phi_o, K_o} \!\!\left[e^{- \frac{4 \pi\lambda}{3}  K_o^{\frac{3}{\sigma}}g(\Theta_o, \Phi_o,\tau)^3} \right]. \tageq 
\label{3d-shannon-outage}
\end{equation}
Using Jensen's inequality we have in general that
\begin{align}
e^{- \frac{4\lambda \pi}{3} \mathbb{E}_{\Phi, \Theta, K}  \left[K^{\frac{3}{\sigma}}g(\Theta,\Phi,\tau)^3\right]} &\leq  \mathbb{E}_{\Theta, \Phi, K} \left[e^{- \frac{4 \pi\lambda}{3} K^{\frac{3}{\sigma}} g(\Theta, \Phi, \tau)^3} \right],\label{eq:jensen-3d}
\end{align}
Comparing the outage probabilities in the proposed Doppler-Shannon framework and in the classical Shannon framework as derived in~\eqref{3d-DS-stochastic-result} and  \eqref{3d-shannon-outage}, respectively, and using (\ref{eq:jensen-3d}) we get
\begin{align} \mathbb{P}\left(\max_i U_i(\Theta_i, \Phi_i, R_i, K_i) {<} \tau \right)  &\leq \mathbb{P}(U(\Theta_o, \Phi_o, K_o, r_o) {\leq} \tau).\nn
\end{align}
\end{proof}

\begin{figure}[t]
    \centering
    \subfloat[][]{
    \includegraphics[width=0.30\textwidth]{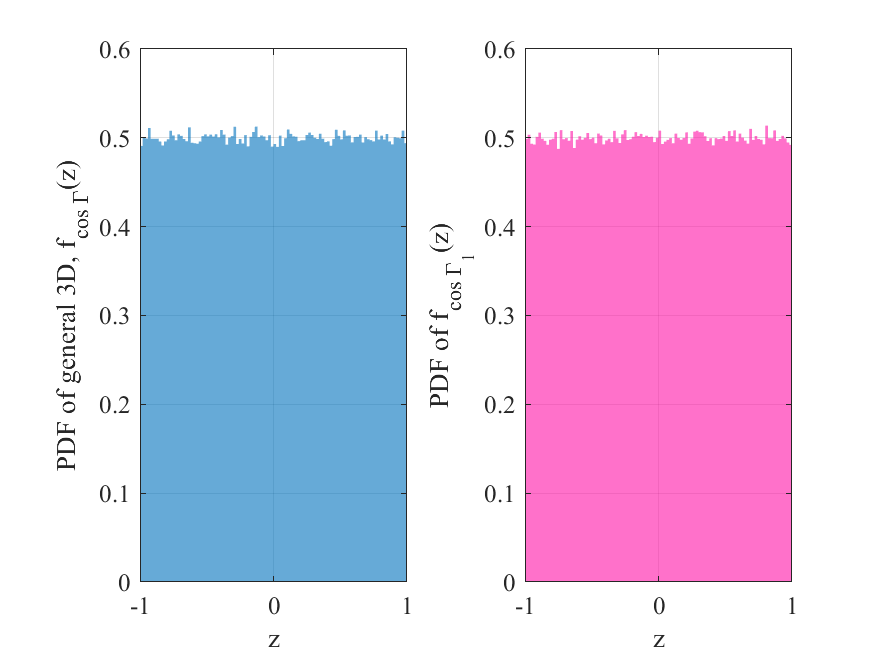}
    }
    \\
    \subfloat[][]{
\includegraphics[width=0.30\textwidth]{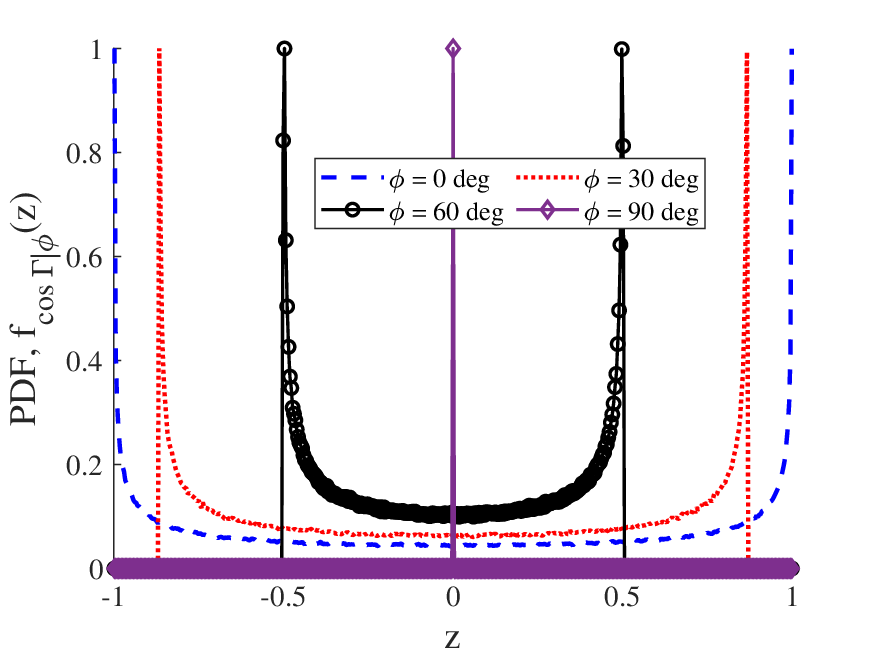}}
    \caption{
    (a) Simulated PDF of $\cos\Gamma$ in generalized 3D case,   \eqref{3d-gamma}, and its comparison with  \eqref{3d-gamma-case2}, (b) Simulation of the conditional PDF of $f_{\cos \Gamma_1|\Phi=\phi}(z)$ in  \eqref{3d-gamma-case2}.} 
    \label{fig:3}
\end{figure}




\section{Results and Discussion}\label{secVI}
    

%

\begin{figure*}[t]
    \centering
    \subfloat[][]{
    \includegraphics[width=0.38\textwidth]{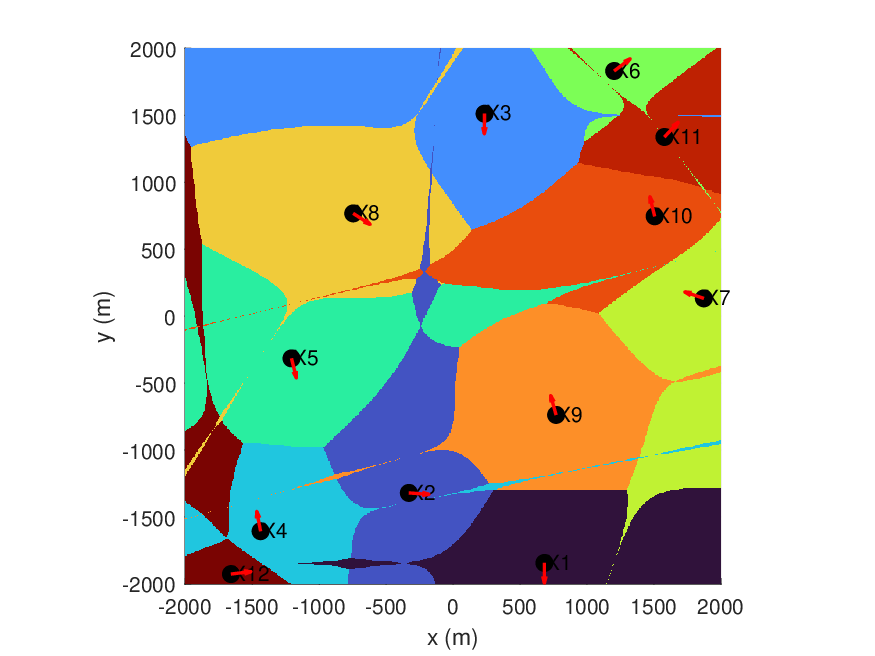}}
    \subfloat[][]{
    \includegraphics[width=0.38\textwidth]{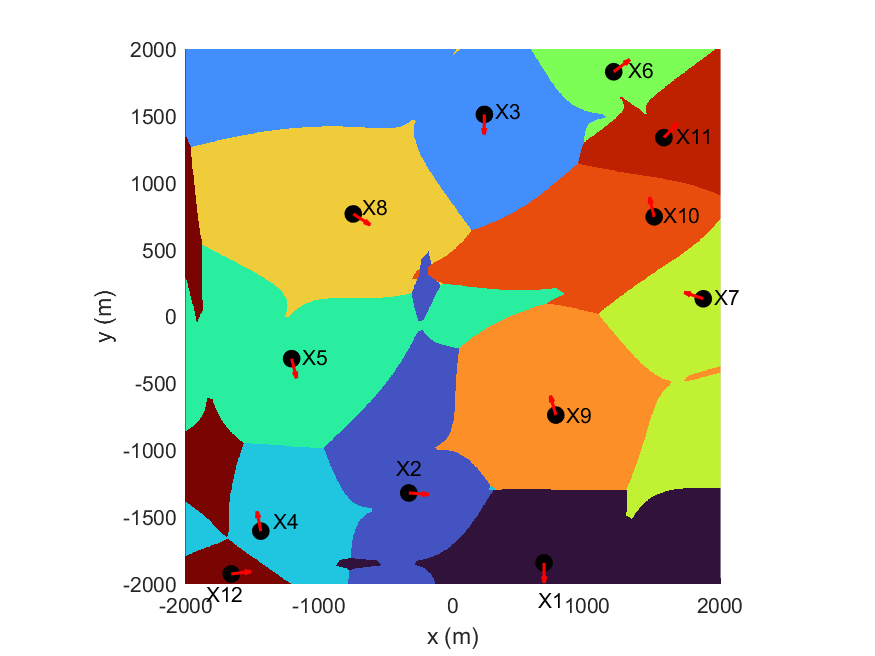}}
    \caption{2D tessellations obtained through  Utility \ref{utility1-proposition} with PLE, $\sigma = 4$ (a) in the absence of limiting factor $\zeta$. Note that the elongated ``rays'' are perpendicular to the DoM of the BS, (b) with limiting factor, $\zeta = 5\times 10^{-3}$. Note that as the limiting factor decreases, the tessellations start becoming more convex eventually looking exactly like the classical Shannon Voronoi.}
    \label{fig:4}
\end{figure*}

\begin{figure*}[t]
    \centering
    \subfloat[][]{
    \includegraphics[width=0.33\textwidth]{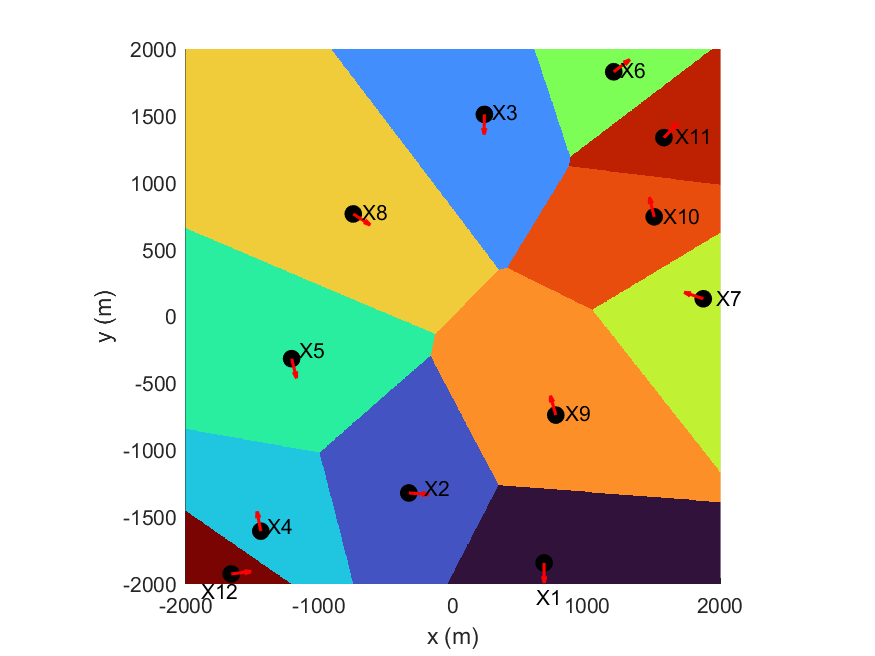}}
    \subfloat[][]{
    \includegraphics[width=0.33\textwidth]{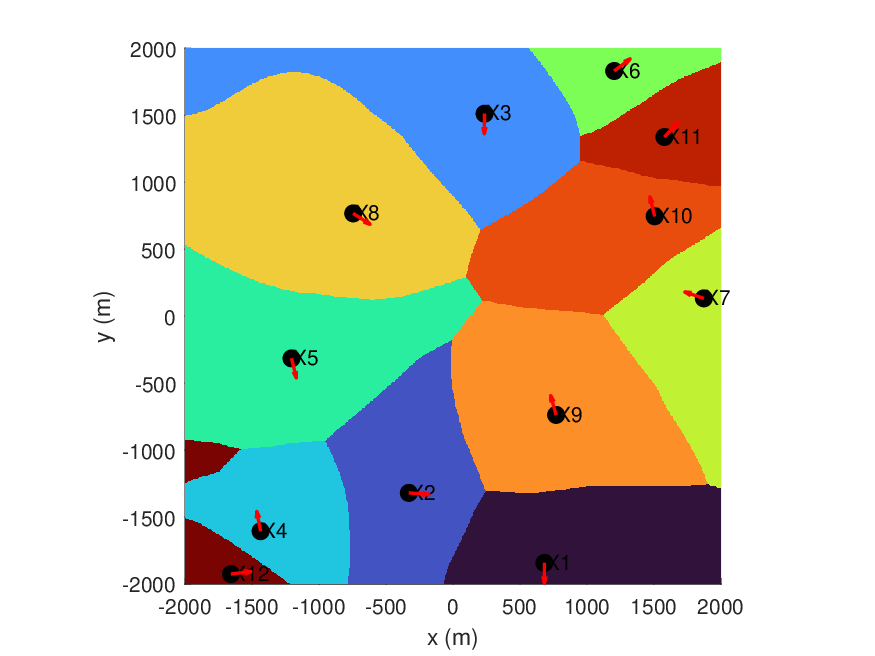}}
    \subfloat[][]{
    \includegraphics[width=0.33\textwidth]{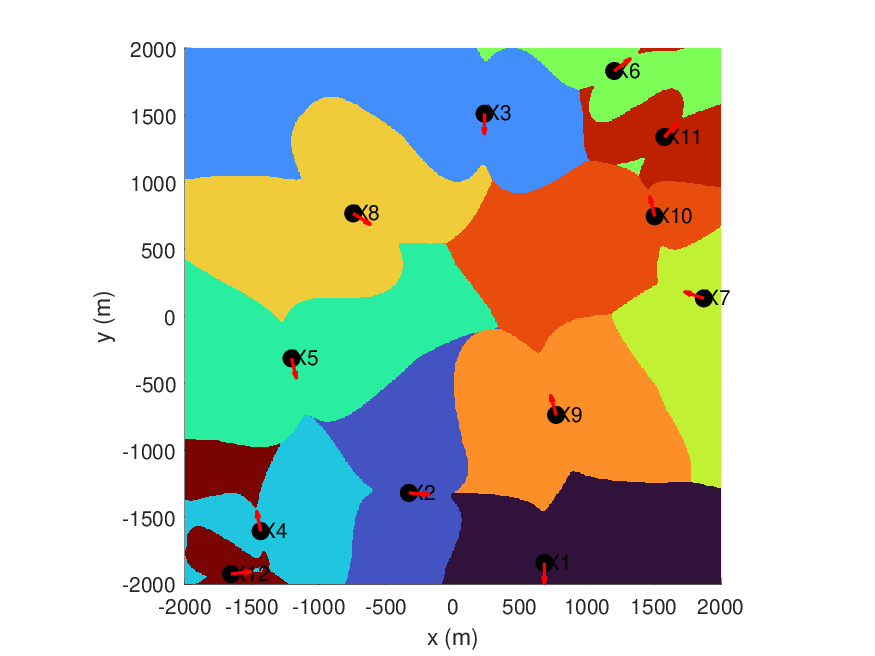}}
    
    \caption{(a) classical Shannon Voronoi tessellations observed through Utility~\ref{utility3-proposition} at $\varepsilon = 10^{-6}$, (b) non-convexity of tessellations obtained through Utility \ref{utility3-proposition}, $\varepsilon=5\times 10^{-5}$, (c) highly non-convex tessellations obtained through Utility~\ref{utility3-proposition} at $\varepsilon = 10^{-4}$ (very inefficient channel estimation interval). 
    }
    \label{fig:5}
\end{figure*}

\begin{figure}[t]
    \centering
    \subfloat[]{
    \includegraphics[width=0.24\textwidth]{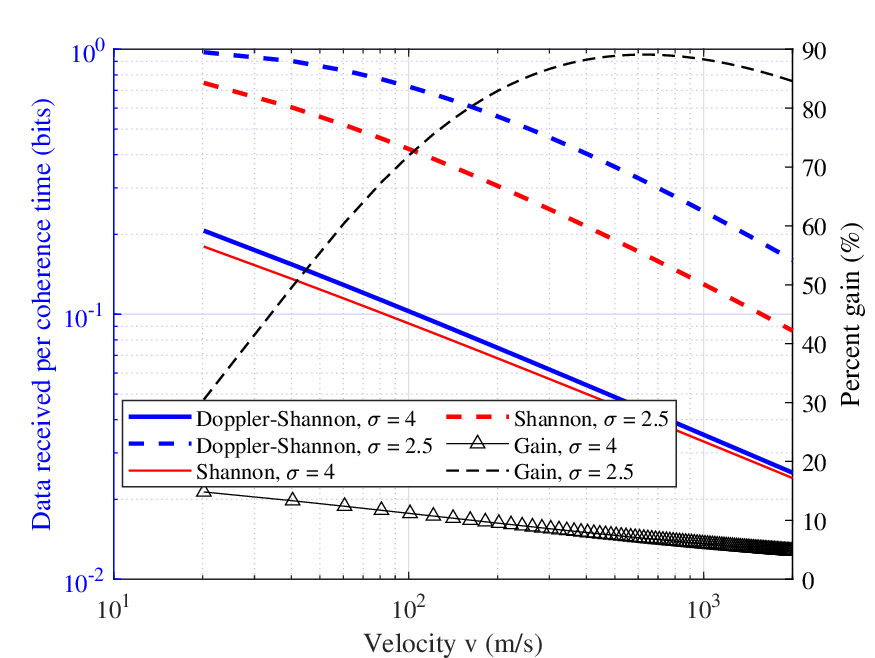}}
    \subfloat[]{
    \includegraphics[width=0.24\textwidth]{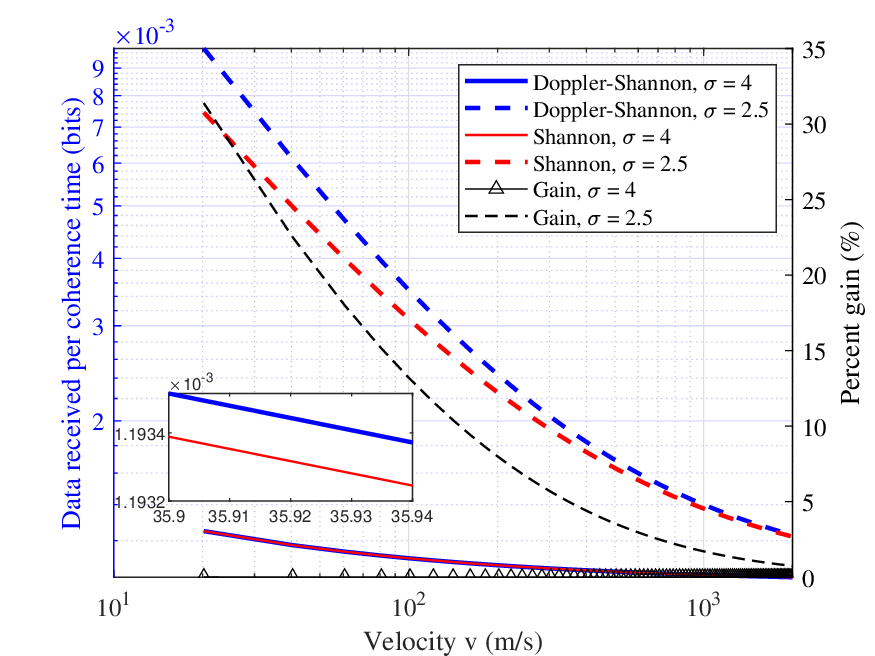}}
    \caption{Performance evaluation of Policy \ref{utility1-proposition} against velocity in a 2D PPP setting with PLE $\sigma = \{2.5, 4\}$ at 
    (a) $\lambda = 10^{-3}$ and (b) $\lambda = 10^{-6} $.
    }
    \label{fig:6}
\end{figure}

\begin{figure*}[t]
    \centering
    \subfloat[]{
    \includegraphics[width=0.25\textwidth]{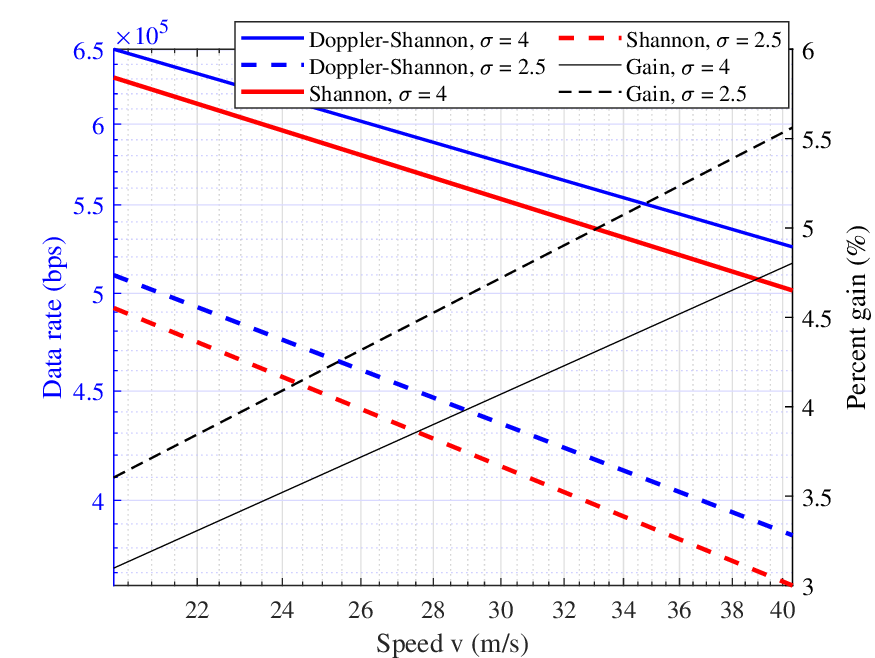}}
    \subfloat[]{
    \includegraphics[width=0.25\textwidth]{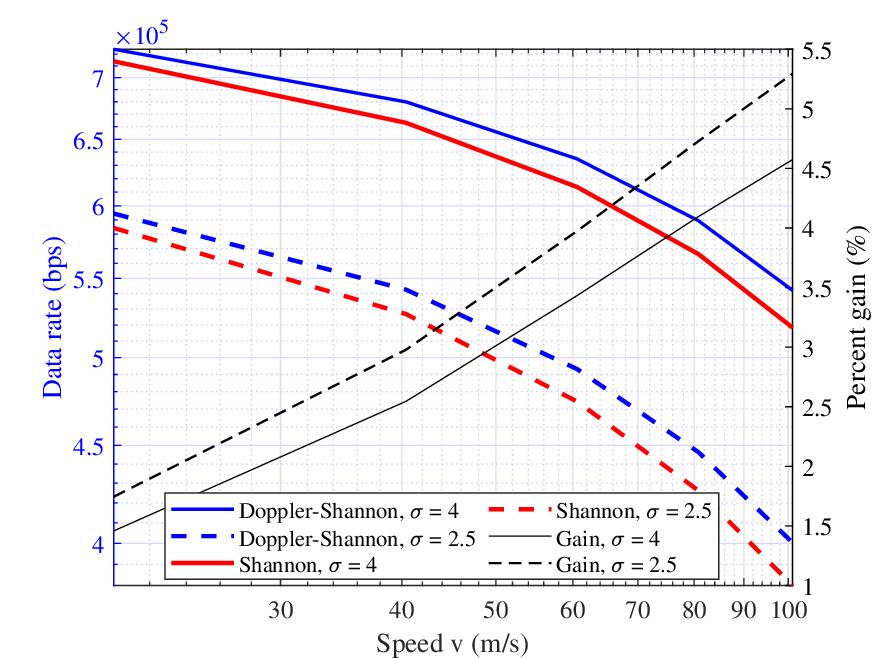}}
    \subfloat[]{
    \includegraphics[width=0.25\textwidth]{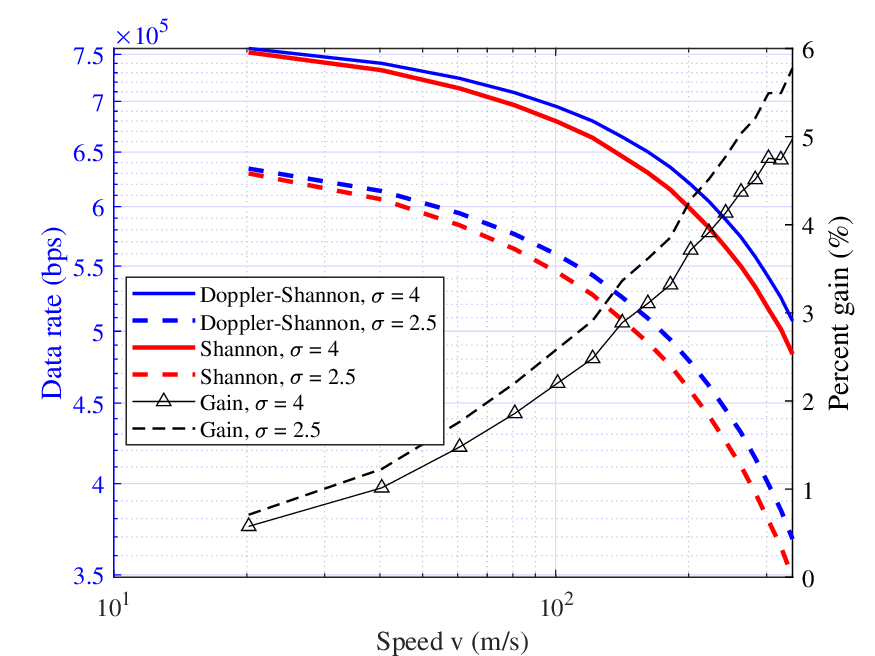}}
    \subfloat[]{
    \includegraphics[width=0.25\textwidth]{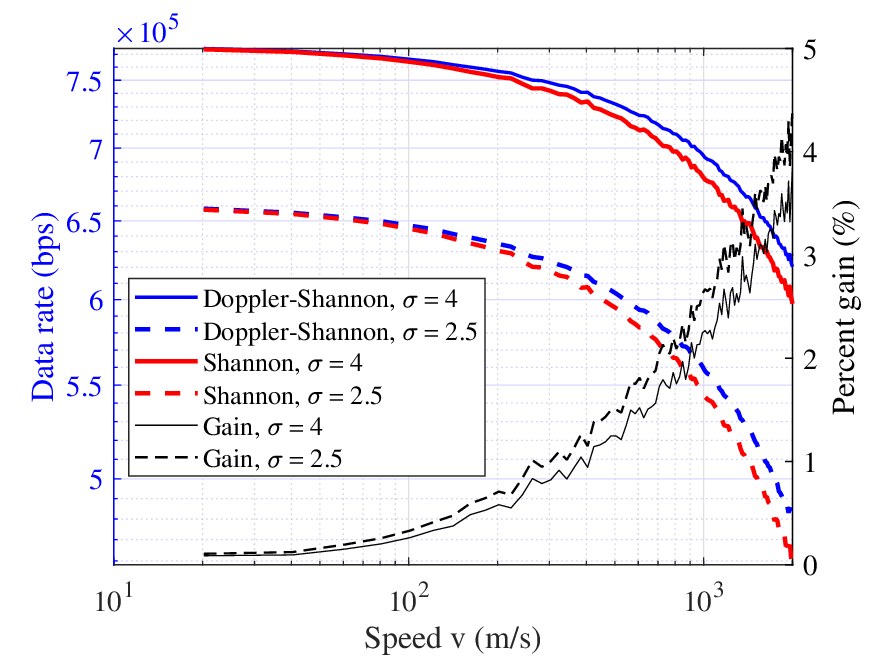}}
    \caption{Performance evaluation of Utility \ref{utility3-proposition} against velocity under 2D PPP setting with PLE $\sigma = \{2.5, 4\}$ at (a) $\varepsilon = 8\times 10 ^{-4}, \lambda = 1$, (b) $\varepsilon = 3 \times 10^{-4}, \lambda = 1$, (c) $\varepsilon = 10^{-4}, \lambda = 1$, (d) $\varepsilon = 10^{-5}, \lambda  = 1$ .}
    \label{fig:8}
\end{figure*}

%
In this section, we present the key findings described analytically previously through simulations. The simulations have been performed in MATLAB-R2024b. The parameter values considered are: $m = 1$, $\Omega = 1$,  $\sigma = 4$, $f_c = 3.5 \mbox{GHz}$, $B = 240 \mbox{KHz}$, $\mu = 0.423$, \ \text{and} \ $c = 3 \times 10^8$. The norm of velocity and the density of the BSs are running variables, and are indicated in the corresponding plots.  

\subsection{Validation of the Doppler distribution results}

From Fig. \ref{fig:3}(a), we validate the analytical proof of Corollary \ref{cor:cor2}. 
It can be clearly observed from Fig. \ref{fig:3}(a) that, the density of the cosine of the Doppler angle  is uniformly distributed in $[-1, 1]$, as claimed in Theorem \ref{theorem-5}. From Fig. \ref{fig:3}(a), we validate the result through \eqref{eq:corollary2}, that in a 3D setting,  when the UE lies on the plane of the BS as in \eqref{3d-gamma-case2}, then the Doppler shift  is statistically equal to the general Doppler shift in \eqref{3d-gamma}.

In Fig. \ref{fig:3}(b), we illustrate the relation between the densities of the Doppler shift in the 3D and 2D settings.  When we condition the cosine component of the Doppler angle, $\cos \Gamma_1$, in \eqref{3d-gamma-case2} with respect to the elevation angle $\phi$, it is observed in Fig. \ref{fig:3}(b) that the resulting densities are arc-sines in nature (as in the 2D scenario, Theorem \ref{theorem1}). Furthermore, for a given elevation angle, it is  inferred that, as the elevation angle $\phi$ increases, the Doppler spread (i.e., the width of the Doppler shift density) decreases. At an elevation of $\phi = 90\deg$, the density becomes an impluse centered around zero. 
\textit{This result also shows the changing nature of the  Doppler spectrum with respect to elevation angle, and thereby the random coherence time intervals (in the time domain) as a UE switches from a mobile node to another. }



\subsection{2D network performance}\label{results-2DT}
In Figs. \ref{fig:4} and \ref{fig:5}, we illustrate the 2D coverage tessellations observed through the novel  Doppler-Shannon association policies proposed in Section \ref{secV}. Note that for better visualization of the tessellation boundaries, we consider unit fading coefficient for all the BSs. In an ideal scenario, the fading will result in blurring the tessellation boundaries.

We start by Fig.~\ref{fig:4}~(a) which depicts the non-convex coverage tessellations obtained for the mobile BSs through Policy \ref{utility1-proposition} in the absence of limiting factor. It can be clearly seen that the tessellations get elongated at a perpendicular angle to the DoM. This is because, as the UE LoS to the BS becomes perpendicular to the BS's DoM (i.e., the BS is moving perpendicular to the UE LoS at that instant), the corresponding Doppler shift is the minimal, resulting in maximal coherence time and thereby maximal utility. 

In the presence of a limiting factor $\zeta$ (\ref{fig:4}(b)), the elongated rays emerging from the tessellations  as seen in  \ref{fig:4}(a) get significantly reduced. It has been observed that decreasing $\zeta$ further, leads to the classical Shannon based Voronoi tessellation by minimizing the Doppler effects on the coverage tessellations. 

The coverage tessellations for Utility \ref{utility3-proposition} are shown in Figs. \ref{fig:5}(a)--(c). It is observed that at $\varepsilon = 10^{-6}$ (i.e., very small channel estimation time interval and hence very efficient channel estimation), the effects of the Doppler shift are not visible on the tessellations. In this case, the tessellations resemble the classical Shannon-Voronoi tessellation. As the channel estimation time increases (i.e., the channel estimation efficiency becomes poorer), the effects of Doppler shift are very clearly visible on the tessellation shapes ($\varepsilon = 5\times 10^{-5}$, Fig. \ref{fig:5}(b)). The tessellations are observed to completely loose their convexity at very poor channel estimation efficiency ($\varepsilon = 10^{-4}$, Fig.~\ref{fig:5}~(c)).

\begin{remark}
    In vehicular networks, the Doppler effect  leads to non-convex coverage tessellations. When coupled with channel estimation, very good estimation (lower estimation time) reduces the effect of the Doppler shift on system performance, while Doppler effect is more prominently visible at poor channel estimations. 
\end{remark}

Through Figs. \ref{fig:6} and \ref{fig:8}, we illustrate the network performance obtained in a 2D PPP setting, against the  association policies proposed in Section V. \textit{In general, for all the policies, the performance obtained through the Doppler-Shannon framework is always higher than or equal to the Shannon framework, as proved analytically in Theorem \ref{theorem:CovP-S-DS}.} Further, the network performance is observed to decrease with increasing velocity norm. Furthermore, the density of points decrease, the Doppler-Shannon and the Shannon frameworks are seen to be converging. This behavior is because with decreasing density of the BS points from the UE, the association Policy does not have much of a selection choice and the effect of Doppler diminishes, resulting in convergence of the plots. 

We have evaluate the network performance with two different PLE, $\sigma = 2.5$ and $\sigma = 4$. Note that a lower  pathloss exponent (here $\sigma = 2.5$) indicates fewer clutter/ obstacles between the UE and BS, thereby lower attenuation of the transmit signal when received at the UE.  
For UE associations based on Policy \ref{utility1-proposition} (arising from Utility \ref{utility1-proposition}, Fig. \ref{fig:6}), it is observed that with a PLE $\sigma = 2.5$, the proposed Doppler-Shannon framework obtains maximum gains up to $90\%$ at $\lambda = 10^{-3}$ (Fig. \ref{fig:6}(a)). 
With a higher PLE $(\sigma = 4)$, the gains obtained from the Doppler-Shannon framework over the Shannon framework reduce to around $14\%$. 
At a reduced $\lambda = 10^{-6}$, the gains obtained with both the PLE decrease in general (Fig. \ref{fig:6}(b)), with the Doppler-Shannon framework gaining up to $35\%$ at $\sigma = 2.5$. In general, it is seen that as the BS density decreases, the Doppler-Shannon and Shannon frameworks converge, in accordance with the analytical results proved through  Jensen's inequality in Theorem \ref{theorem:CovP-S-DS}.



\begin{figure}[t]
    \centering
    \subfloat[]{
    \includegraphics[width=0.25\textwidth]{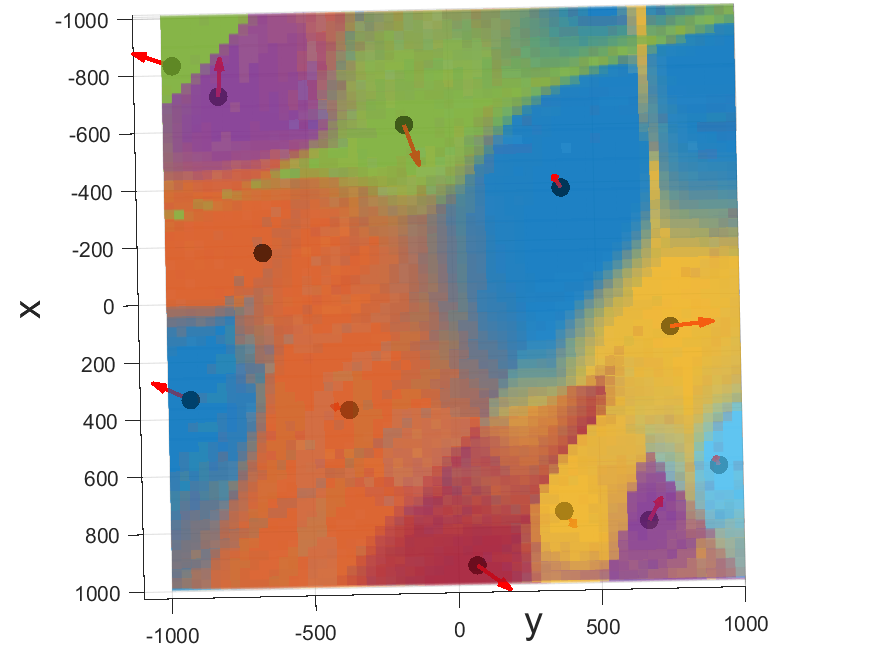}}
    \hspace{-0.35in}
    \subfloat[]{
    \includegraphics[width=0.25\textwidth]{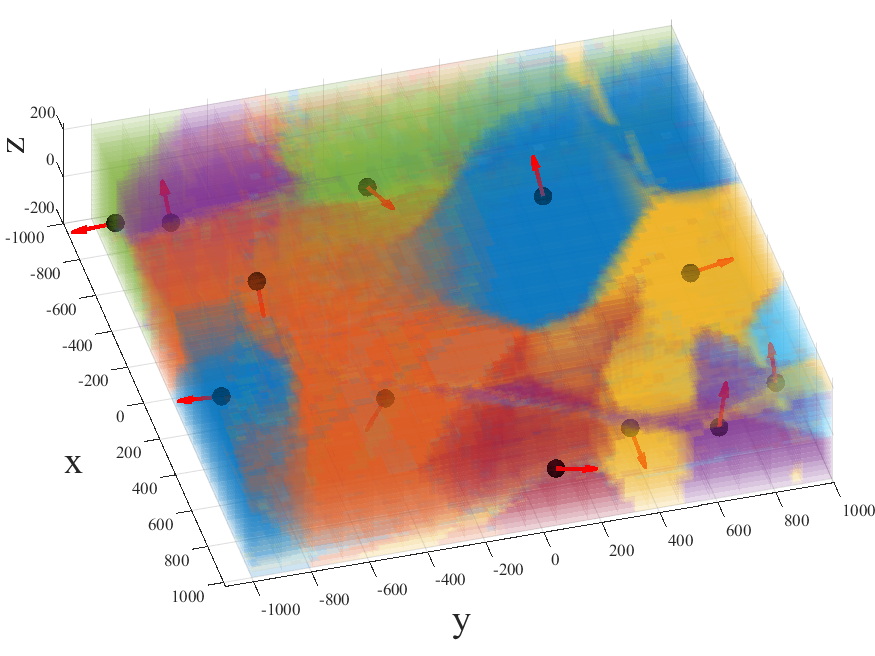}}
    \caption{3D tesselations originating from associations through Policy \ref{utility1-proposition} (a) top view, (b) side view. }
    \label{fig:9}
\end{figure}

\begin{figure}[t]
    \centering
    \subfloat[]{
\includegraphics[width=0.25\textwidth]{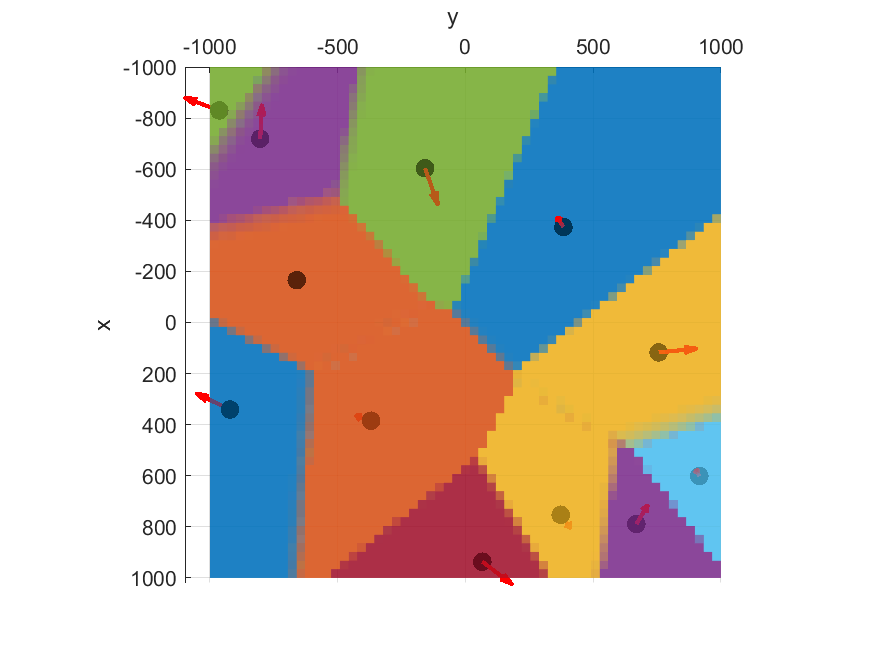}}
    \subfloat[]{    \includegraphics[width=0.25\textwidth]{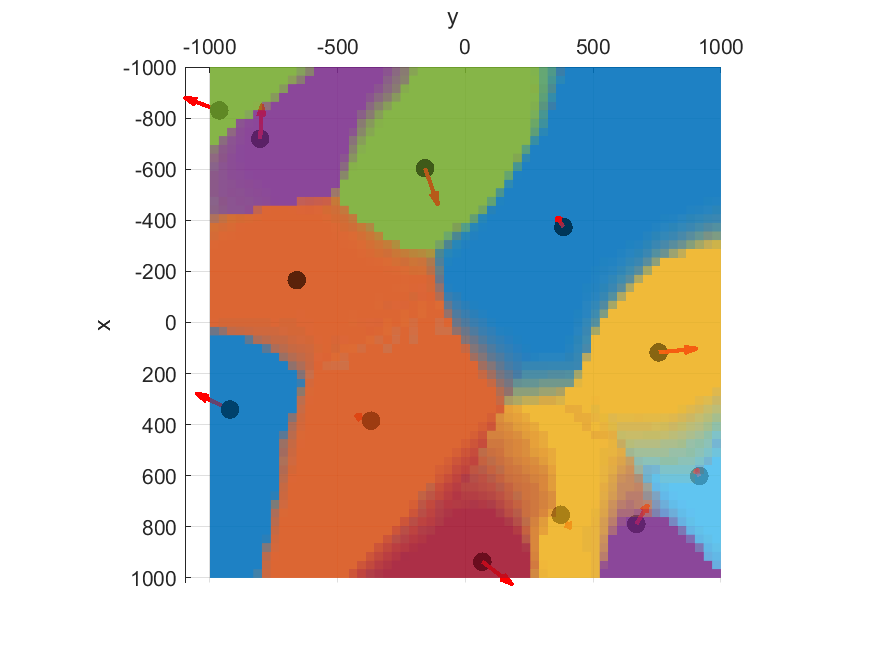}}
    \caption{3D tesselations (top views) originating from associations through Utility \ref{utility3-proposition} (a) very good channel estimation (effect of Doppler not visible), (b) poor channel estimation (effect of Doppler on the tessellations is clearly visible. }
    \label{fig:11}
\end{figure}

\begin{figure}[t]
    \centering
    \subfloat[]{
    \includegraphics[width=0.24\textwidth]{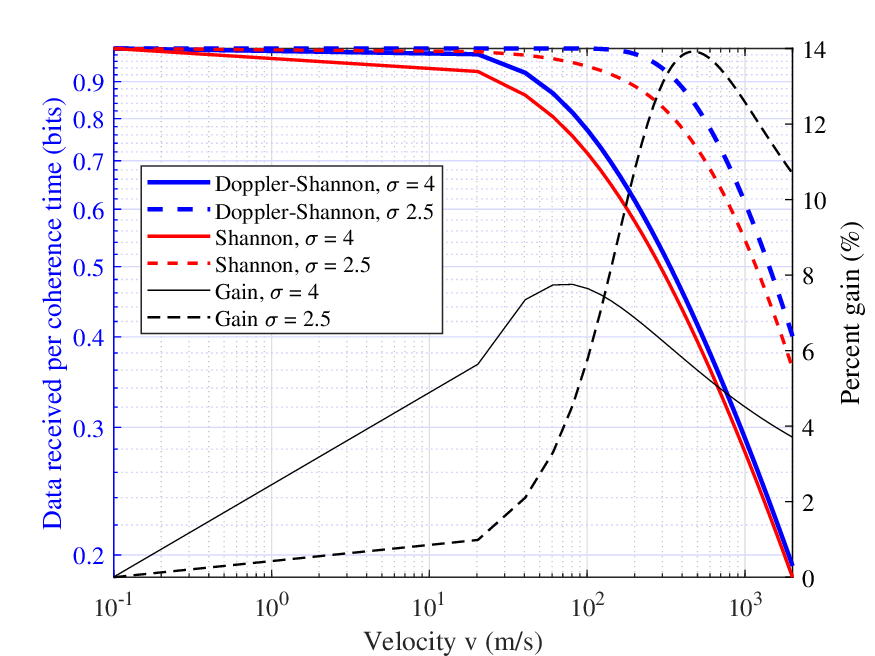}}
    \subfloat[]{
    \includegraphics[width=0.24\textwidth]{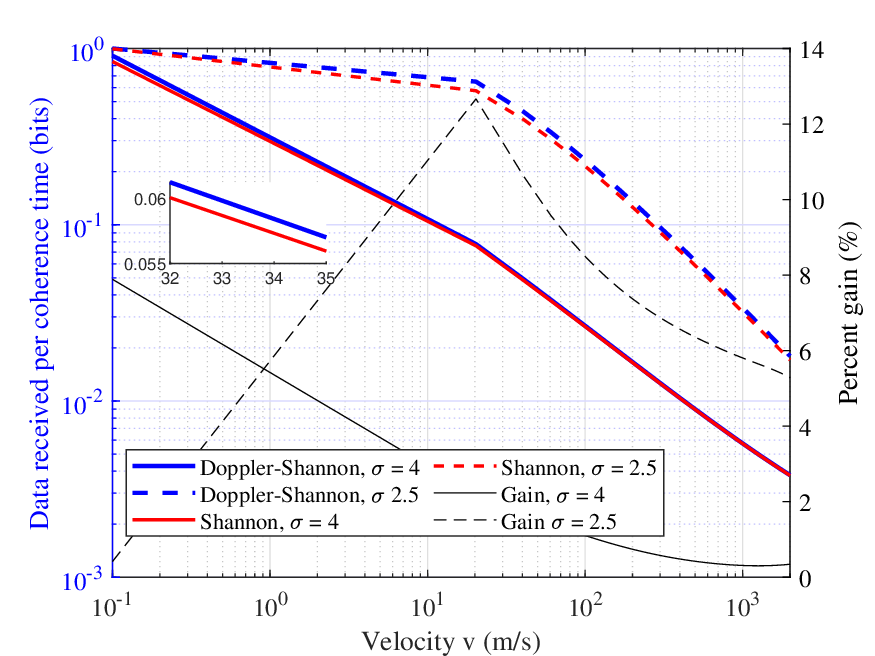}}
    \caption{Performance evaluation of Policy \ref{utility1-proposition} against velocity under 3D PPP setting with PLE $\sigma = \{2.5, 4\}$ at (a) $\lambda = 10^{-2}$ 
    and (b) $\lambda = 10^{-4}$.}
    \label{fig:12}
\end{figure}

\begin{figure*}[t]
    \centering
    \subfloat[]{
\includegraphics[width=0.30\textwidth]{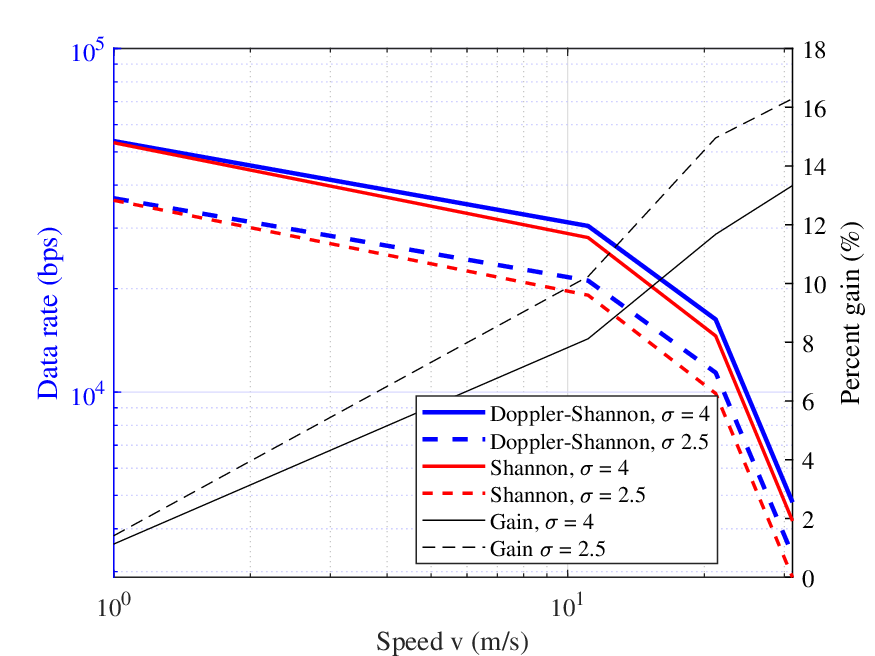}}
    \subfloat[]{
\includegraphics[width=0.30\textwidth]{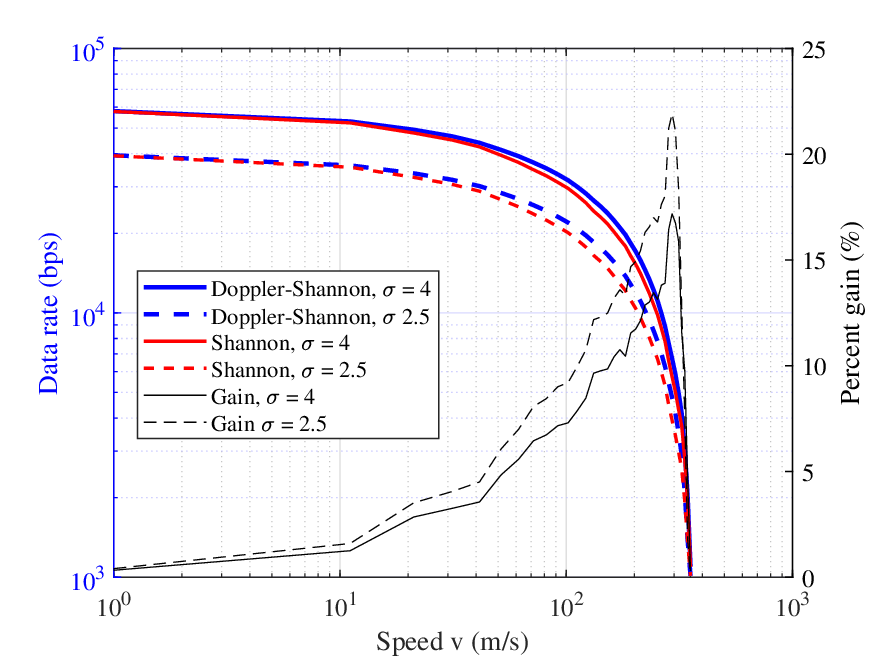}}
    \subfloat[]{
\includegraphics[width=0.30\textwidth]{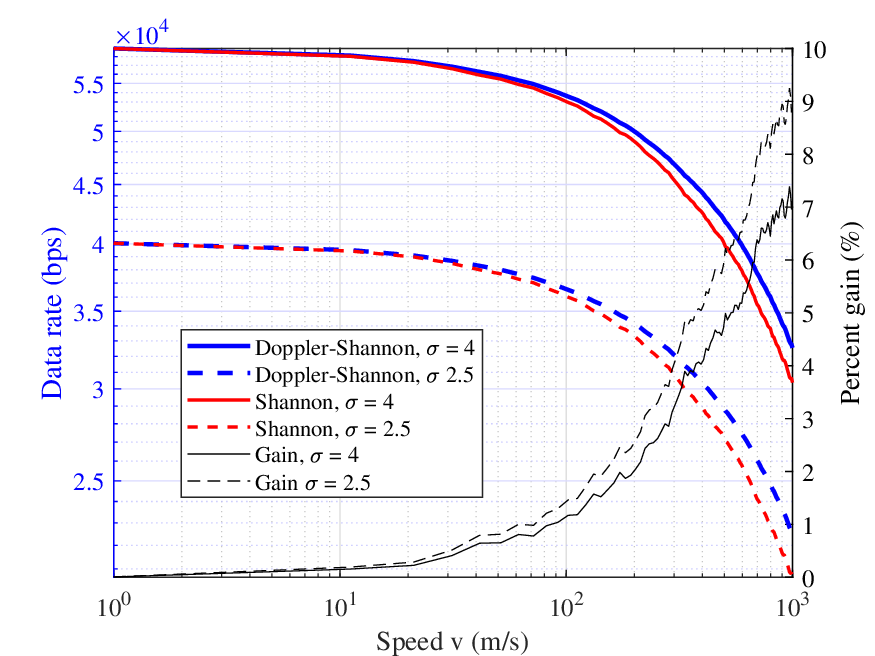}}
    \caption{Performance evaluation of Utility \ref{utility3-proposition} against velocity under 3D PPP setting with PLE $\sigma = \{2.5, 4\}$ at (a) $\varepsilon = 10^{-3}, \lambda = 1 $, (b) $\varepsilon= 10^{-4}, \lambda = 1$, (c) $\varepsilon = 10^{-5}, \lambda = 1$.}
    \label{fig:14}
\end{figure*}

With Utility \ref{utility3-proposition} (Fig. \ref{fig:8}), the network performance is plotted not only w.r.t. the BS densities but also the channel estimation time interval. 
It is observed that, with a PLE of $\sigma = 2.5$, the Doppler-Shannon framework obtains higher gains by around  1\% as compared to that obtained from $\sigma = 4$.  
Overall, the Doppler-Shannon framework is observed to provide rate gains by around $5\%$ over the Shannon framework.

\textit{Interestingly, it is observed that the quality of channel estimation has an effect on the magnitude of velocity that the system can support.} For instance, at very poor channel estimation efficiency (higher channel time) as in Fig. \ref{fig:8}(a), the system is observed to support only up to 40 mps. As the channel estimation quality improves, the system can be observed to support higher mobility speeds (around 2000 mps, i.e., speed of LEO satellite,  Fig. \ref{fig:8}(d)). Further, at very efficient channel estimation, the performance of the Doppler-Shannon and the Shannon frameworks are seen to converge (Fig. \ref{fig:8}(d)). This result also validates the observation previously made on the coverage tessellations in Fig. \ref{fig:5}. 
%
%


\subsection{3D network performance}

Figs. \ref{fig:9} and \ref{fig:11} illustrate the coverage tessellations in 3D.  They can be observed in parallel with the 2D tessellations depicted in Figs. \ref{fig:4}--\ref{fig:5}. For instance, the top view of the 3D tesselation originating from Policy \ref{utility1-proposition} in Fig. \ref{fig:9}(a) is conceptually similar to Fig. \ref{fig:4}(a), wherein the coverage tessellations are seen to elongate infinitely at a direction perpendicular to the direction of motion of the BS.  Fig. \ref{fig:11}(a) illustrates the ``Shannon equivalent'' 3D convex tessellation observed at very good channel estimation quality (as also in Fig. \ref{fig:5}(a)), obtained through Utility \ref{utility3-proposition}.   These tessellations are observed to lose convexity (Fig. \ref{fig:11}(b)) at poorer channel estimation quality, as also inferred through Fig. \ref{fig:5}(b)-(c). 

The network performance evaluation for the 3D setting with the Doppler-Shannon association and the classical Shannon association, against the proposed association policies, is illustrated in Figs. \ref{fig:12}--\ref{fig:14}. The nature of the plots in the 3D PPP setting, 
against the respective association policies, are observed to be similar to their 2D equivalents.

The performance of Policy \ref{utility1-proposition} (Fig. \ref{fig:12}) is observed to be decreasing with BS velocity. The gain is observed to be maximum of around 14\%, with a PLE $\sigma = 2.5$, at a  moderate BS density of $\lambda=10^{-2}$ (Fig. \ref{fig:12}(a)). The corresponding maximum gain, with a PLE $\sigma = 4$, is observed to be around $7\%$. 
In general, the gains obtained through the Doppler-Shannon framework are observed to decrease with decreasing $\lambda$ (Fig. \ref{fig:12}(b)).  The network performance obtained through the Doppler-Shannon and the Shannon policies are observed to converge with decreasing $\lambda$, as proved in Theorem \ref{theorem-9}.


The Doppler-Shannon network performance of Utility \ref{utility3-proposition} (Fig. \ref{fig:14}) in 3D is observed to provide  higher gains (peaking around 22\% in Fig. \ref{fig:14}(b) with PLE $\sigma = 2.5$) as compared to the 2D setting, where the maximum gains obtained were about of 5-6\%. With a higher PLE $\sigma = 4$, the gains reduce to around 16\%. 
In general, the gains in the Doppler-Shannon based network performance is seen to improve with lower PLE. 
The variation of velocity support of the system with channel estimation quality is also clearly inferred from the plots in Fig. \ref{fig:14}. For instance, poor channel estimation quality supports speeds in the order of $10^1$ (Fig. \ref{fig:14} (a)), while as the channel estimation quality improves the system's ability to support high mobility in order of $10^3$ (use case of LEO satellites, Fig. \ref{fig:14}(c)) is clearly noted.  
 In general, the Shannon and Doppler-Shannon frameworks are seen to converge when decreasing BS densities, as observed in the  2D plots as well.



\section{Concluding Remarks and Future Directions}
In this work, we have designed and evaluated the performance of Doppler shift aware vehicular networks wherein the network UEs experience Doppler shift due to high mobility of the BSs. The Doppler shift results in the fading channel exhibiting high temporal variability, the effects on the physical layer network performance of which have been modeled through the coherence time.
We have first presented the statistical characterization of the Doppler shift in the 2D and then in the 3D PPP setting. 
Next, Doppler shift aware utilities have been proposed to model the Shannon rate achievable in time-selective fading links. Stochastic geometry tools have been utilized to study the outage and coverage probability for each of the utility proposed.  It has been shown that in high mobility environments, the optimal BS association is no longer the closest Euclidean BS. Rather, we demonstrate that the effect of Doppler shift results in non-convex coverage tessellations for the BSs. Our simulation results have shown the need for a Doppler correction in vehicular networks, to complement the Shannon performance. The network performance achieved through the proposed Doppler-Shannon based association are proved, analytically as well as in simulated results, to be superior to the classical Shannon framework. 
The work is expected to be a first step aiming to capture the temporal fading channel effects on the network performance. 

We would like to highlight that the current  work does not study the system as a time-space dynamical system. Rather, the analysis is based on the spatial statistics at a snapshot of time. Future research works are expected to centered around the study the time dynamical systems, extending the model to the spherical case, and the study of the signal-to-interference-plus-noise-ratio (SINR) in Doppler scenarios.  
%
Furthermore, as an independent remark, we would like to highlight that when the angles have a joint distribution with a dependent marginal, then the distribution of the Doppler shift in 3D is not uniform as in Theorem \ref{theorem-3}. As a proof of concept, we have provided an example case of LoS dependent motion in Appendix 
J. 
A more detailed performance evaluation of such LoS dependent motion models are left for future study.

\appendix

\subsection{Proof of Theorem \ref{theorem1}}\label{app-0}
\begin{proof}
The density of $\Gamma$ is
\begin{align}
f_{\Gamma}(\gamma){=} \!\!\int_0^{2\pi}\!\!\!\!\!f_\Theta((\gamma{+}a) \bmod2\pi) f_\alpha(a){\rm d}a
{=}\!\!\int_0^{2\pi}\!\!\! \frac{1}{2\pi} f_\alpha(a){\rm d}a {=} \frac{1}{2\pi}.\nn
\end{align}
 We have the Doppler shift 
    $D  {=} \frac{v\cos (\Theta - \alpha)}{\rho} \stackrel{d}{\sim} \frac{v\cos \Theta}{\rho}.$
Here, $\cos \Theta\!\! \stackrel{d}{\sim}\!\! \mbox{U}[-1, 1]$ resulting in $\frac{v\cos \Theta}{\rho} \stackrel{d}{\sim} \mbox{U}\left[-\frac{v}{\rho}, \frac{v}{\rho} \right]$. Furthermore, $f_\Theta(\theta) {=} \frac{1}{2\pi}$. 
Using the transformation of random variables, $\Theta {=} \arccos{\left(\frac{D  \rho}{v}\right)}$, the PDF of $D$ is
\begin{equation}
\begin{aligned}\label{2D-doppler-shift}
    f_{D}(x) &{=} 2 \times f_\Theta \left( \arccos{\left(\frac{x \rho}{v}\right)} \right) \left| \frac{\dd \theta}{\dd x} \right| 
    \\
        &= \frac{1}{\pi}\left(\left(v/\rho\right)^2  - x^2 \right)^{-1/2}, 
    \end{aligned}
    \end{equation}
for all $x {\in} \left( -v/\rho,v/\rho \right)$. The corresponding CDF of the Doppler shift is
    \begin{align}\label{2d-cdf}
    F_D(x) &{=} \int_{-v/ \rho}^{x} \frac{1}{\pi} \left(\left(v/\rho\right)^2  {-} x^2\right)^{-1/2} \dd x, 
    \\
    &{=} \frac{1}{\pi} \left( \sin^{-1} \left( \frac{x}{v/ \rho} \right) \right) \Bigg|_{x = -\frac{v}{\rho}}^{x}\!\!\!\! \nn
    {=} \frac{1}{\pi}\sin^{-1}\frac{x \rho}{v} {+} \frac{1}{2}.
    \end{align}
\end{proof}


\subsection{Proof of Lemma~\ref{lemma:tc-pdf}}~\label{Tc-pdf}
\begin{proof}
The PDF and CDF of the Doppler in 2D setting has been derived through \eqref{2D-doppler-shift} and \eqref{2d-cdf}, respectively. We derive the distribution of the magnitude of the Doppler shift below as
%
\begin{align}
    F_{|D|}(w) &= F_D(w) - F_D(-w) = \frac{2}{\pi}\sin^{-1}\frac{w \rho}{v},\nn
\end{align}
for $w \in [0, v/{\rho})$. Hence, the density of $|D|$ is
%
$f_{|D|}(w) = \frac{2}{\pi} \left(\left(v/{\rho}\right)^2 - w^2\right)^{-\half},$ 
%
for $w \in [0, v/{\rho})$. 
Below, the density of the length of the coherence time interval $T_c = \frac{\mu}{|D|}$ is derived using the transformation of random variables $w=\mu/y$, as  
\begin{align}\label{coherence-pdf}
    f_{T_c}(y) &{=} f_{|D|}\left(\frac{\mu}{y}\right) \left| \frac{\dd w}{dy} \right| {=} \frac{2\mu}{\pi y^2} \!\!\left(\left(\frac{v}{\rho}\right)^2 \!\!{-} \!\left( \frac{\mu}{y} \right)^2\right)^{-\half}\!\!\!\!\!,\nn\\
    &=  \frac{2}{\pi y} \!\!\left(\left(v y/\rho \mu \right)^2 \!\!{-} 1\right)^{-\half}, \forall \ y \in \left[\mu \rho/v, \infty \right).
\end{align}
for all $y \in \left[\mu \rho/v, \infty \right)$.
\end{proof}


\subsection{Proof of Lemma \ref{lemma-1}}\label{app-lemma1a}
To have the BS locations uniformly distributed on the 3D space,  given the LoS distance $r$, the locations must be uniformly distributed on the sphere of radius $r$. We use the construction of~\cite{pitman2012archimedes} that, given LoS distance $r$, the location of the BS is
\begin{align}
    \underbar{$r$} &=r\left[\sqrt{1-V^2}\cos\Theta, \sqrt{1-V^2}\sin\Theta, V\right]^T,\label{eq:pitman}
\end{align}
where $V\stackrel{d}{\sim}\mbox{U}[-1,1]$ and $\Theta\stackrel{d}{\sim} \mbox{U}[0,2\pi)$. The coordinate representation in (\ref{eq:pitman}) is equivalent to our convention ($\underbar{$r$} {=} r[\cos \Phi  \cos\Theta\; \cos \Phi  \sin \Theta\; \sin \Phi]^T$) under the change of variable $V=\sin \Phi$, where $\Phi$ denotes the elevation angle of the LoS. Using $V\stackrel{d}{\sim}\mbox{U}[-1,1]$, we have
\begin{align}
\P(\Phi\leq \phi)&= \P(V\leq \sin\phi)= \half \left(\sin\phi+1\right),
\end{align}
for $\phi\in [-\pi/2,\pi/2]$. This gives the density of $\Phi$ to be $f_{\Phi}(\phi)=\half \cos \phi$.

\subsection{Proof of Theorem \ref{theorem-3}}\label{app-1}

\begin{proof}
Let us recall that, from~\eqref{3d-gamma}, the radial component of the Doppler angle is given by   
\begin{equation}\label{3d-gamma3}
\begin{aligned}
    \cos \Gamma &= \cos \alpha \cos \Phi \cos( \Theta -\beta) + \sin \alpha \sin \Phi
    \\
   & \stackrel{d}{\sim} \cos \alpha \cos \Phi \cos \Theta + \sin \alpha \sin \Phi.
\end{aligned}
\end{equation}
The simplification follows from the fact that $\Theta - \beta \stackrel{d}{\sim} \Theta$, similarly to Theorem~\ref{theorem1}, given the azimuth $\beta$. 
%
%
Let, $\cos \Gamma {=} Z$ and $ \cos \Theta {=} Y$. Then by re-arranging the terms in \eqref{3d-gamma3}, we get
\begin{equation}
Y{=} \left(Z{-}\sin \a \sin \Phi\right)\slash \left(\cos \a \cos \Phi\right).
\label{eq:scaledZ}
\end{equation}
The distribution of $Y {=} \cos\Theta$, for $\Theta \stackrel{d}{\sim} \mbox{U}[0, 2 \pi)$ is computed using the change of variables $y=\cos\theta$ as 
\begin{equation}\label{pdf-y-cos-thetha}
    \begin{aligned}
        f_Y(y) {=} 2f_{\Theta}(\cos^{-1}{y})\left\vert\frac{\dd \theta}{\dd y}\right\vert {=} 2 {\times} \frac{1}{2 \pi} \left| \frac{-1}{\sqrt{1 - y^2}} \right| 
        {=} \frac{1}{\pi \sqrt{1 {-} y^2}},
    \end{aligned}
\end{equation}
%
%
for all $y \in [-1, 1]$. 
Using the relation (\ref{eq:scaledZ}) and the density of $Y$, we can write the conditional probability density of $Z$ given $\alpha = a, \Phi =  \phi$, as
\begin{align}
f_{Z\vert \alpha, \Phi}(z\vert a, \phi) &{=} \frac{1}{\cos a \cos \phi}\frac{1}{\pi} \left(1{-} \left(\frac{z{-}\sin a \sin \phi}{\cos a \cos \phi}\right)^2\right)^{-\half}\nn\\
&{=}\frac{1}{\pi} \left(\cos^2 a \cos^2 \phi{-}\left(z{-}\sin a \sin \phi\right)^2\right)^{-\half}\!\!\!\!\!\!\!,
\label{eq:GC0}
\end{align}
where both $\cos a, \cos \phi$ are positive on $[-\frac{\pi}{2}, \frac{\pi}{2}]$. The conditional probability density is well defined if $\vert z-\sin a \sin \phi\vert \leq \cos a \cos \phi$. On the other hand, after some rearrangements $f_{Z\vert \alpha, \Phi}(z\vert a, \phi)$ in (\ref{eq:GC0}) can be written as
\begin{align}
f_{Z\vert \alpha, \phi}(z\vert a, \phi)
%
%
%
%
%
%
&{=}\frac{1}{\pi}\!\! \left((1{-}z^2)\cos^2 \phi{-}(\sin a {-} z\sin \phi)^2\right)^{-\half}\nn\\
&= \frac{(1-z^2)^{-\half}}{\pi\cos \phi}  \left(1-\frac{(\sin a - z \sin \phi)^2}{(1-z^2)\cos^2\phi }\right)^{-\half}.
\label{eq:GC1}
\end{align}
The conditional probability density (\ref{eq:GC1}) is well defined if $\vert \sin a - z\sin \phi\vert \leq (1-z^2)^\half \cos \phi$. The probability density of $Z$ turns out to be the following, after averaging the conditional probability density in (\ref{eq:GC1}) with respect to $\alpha \ \text{and} \ \Phi$, using their densities $f_{\Phi}(\phi)=\half \cos \phi$ and $f_{\a}(a)=\half \cos a$
\begin{align}
&f_Z(z)=\int_{-\pi/2}^{\pi/2} \int_{-\pi/2}^{\pi/2}  f_{Z\vert \alpha, \phi}(z\vert a, \phi)\,  f_{\Phi}(\phi)\, f_{\a}(a)\,  \dd  a\, \dd \phi\nn\\
&=\int_{-\pi/2}^{\pi/2} \int_{-\pi/2}^{\pi/2}  f_{Z\vert \alpha, \phi}(z\vert a, \phi)\,  \frac{1}{4}\cos \phi\, \cos a\,  \dd a\, \dd \phi\nn\\
%
%
%
&\stackrel{(\ref{eq:GC1})}{=}\frac{1}{4\pi}(1-z^2)^{-\half}\int_{-\pi/2}^{\pi/2} \int_{-\pi/2}^{\pi/2} \left(1{-}\frac{(\sin a - z \sin \phi)^2}{(1-z^2)\cos^2\phi }\right)^{-\half}\nn\\ 
&\hspace{1.8in} \times  \cos a \, \dd a\, \dd \phi.
\label{eq:GC2}
\end{align}
Using the change of variable $\sin a=t$, we get
\begin{equation}
\resizebox{\hsize}{!}{$f_Z(z){=} \frac{(1{-}z^2)^{-\half}}{4\pi} \int_{-\pi/2}^{\pi/2} \int_{-1}^{1}  \left(1{-}\frac{(t{-}z \sin \phi)^2}{(1{-}z^2)\cos^2\phi }\right)^{-\half} \dd t\, \dd \phi.$}
\label{eq:GC3}
\end{equation}
The integrand in the last integral is well defined only on the region 
\begin{align}
R_{z,\phi}&:=\left\{t\in \R:\vert t-z \sin \phi\vert\leq (1-z^2)^\half\cos\phi \right\},
\end{align}
since $(1{-}z^2)^\half\cos\phi{\geq} 0$, for any $\phi{\in} [-\frac{\pi}{2},\frac{\pi}{2}]$ and $z{\in} [-1,1]$. It can be shown that $R_{z,\phi}{\subset} [-1,1]$, since $z\sin\phi{+}(1{-}z^2)^\half\cos\phi {\leq} 1$ and $z\sin\phi{-}(1{-}z^2)^\half\cos\phi {\geq} -1$. For any $z\in [-1,1]$, take $z{=}\sin\psi$, for some $\psi\in [-\pi/2,\pi/2]$. Then we have that
\begin{align}
z\sin\phi{+}(1{-}z^2)^\half\cos\phi &{=} \sin\psi\sin\phi{+}\cos\psi\cos\phi\nn\\
&{=} \cos(\psi{-}\phi)\leq 1, \nn
\end{align}
\begin{align}
z\sin\phi{-}(1{-}z^2)^\half\cos\phi&{=} \cos(\psi{+}\phi)\geq -1, \notag
\\
\hspace{-1in} \text{and,} \  z\sin\phi{+}(1{-}z^2)^\half\cos\phi&{-}(z\sin\phi{-}(1{-}z^2)^\half\cos\phi) \notag\\
&\hspace{-0.5in} {=}2(1{-}z^2)^\half\cos\phi\geq 0,
\end{align}
for all $\phi \in[-\frac{\pi}{2}, \frac{\pi}{2}]$. Hence the region $R_{z,\phi}$ is an interval on $[-1,1]$ and
\[
R_{z,\phi}=\left[z\sin\phi{+}(1{-}z^2)^\half\cos\phi, z\sin\phi{-}(1{-}z^2)^\half\cos\phi\right].
\]
Then simplifying the inner integral in (\ref{eq:GC3}) while performing the integration with respect to $t$ in a smaller region $R_{z,\phi}$, we obtain 
\begin{align}
\lefteqn{\int_{R_{z,\phi}}  \left(1-\frac{(t-z \sin \phi)^2}{(1-z^2)\cos^2\phi }\right)^{-\half} \, \dd t}\nn\\
&= \arcsin\left(\frac{t-z \sin \phi}{(1-z^2)^\half\cos\phi}\right) \Bigg\vert_{z\sin\phi-(1-z^2)^\half\cos\phi }^{z\sin\phi+(1-z^2)^\half\cos\phi} \hspace{-0.4in}\times (1-z^2)^\half\cos \phi\nn\\
%
%
&=\pi (1-z^2)^\half\cos \phi.
\label{pi-phi1}
\end{align}
To compute the density of $Z$, we substitute from~\eqref{pi-phi1} to (\ref{eq:GC3}) to get
\begin{align}
f_Z(z) &=\frac{1}{4\pi}\frac{1}{(1-z^2)^\half} \int_{-\pi/2}^{\pi/2} \pi (1-z^2)^\half\cos \phi\, \dd \phi\nn\\
%
%
%
&= \frac{1}{4} \int_{-\pi/2}^{\pi/2} \cos \phi \, \dd \phi= \half.
\label{eq:GC4}
\end{align}
This  shows that the random variable $Z=\cos\Gamma$, involved in the Doppler shift, to be distributed as $\mbox{U}[-1,1]$. 
Further, the Doppler shift $D = \frac{v}{\rho} \cos \Gamma$ is distributed as $\mbox{U}[-\frac{v}{\rho}, \frac{v}{\rho}]$ with PDF $f_D(x) = \frac{\rho}{2v}$, for $x\in [-\frac{v}{\rho}, \frac{v}{\rho}]$.
\end{proof}


\subsection{Proof of Lemma \ref{coh-time-3d}}\label{proof-coh-time-3d}
From Theorem \ref{theorem-3} we have, $D \stackrel{d}{\sim}\mbox{U}[\frac{v}{\rho}, \frac{v}{\rho}], i.e, f_D(x) = \frac{\rho}{2v}$. Hence, $f_{|D|}(w) = \frac{\rho}{v}.$ Further, the coherence time is defined as, $T_c = \frac{\mu}{|D|}$. Since, $|D| \in [0, \frac{v}{\rho}]$, hence $T_c \in [\frac{\mu v}{\rho}, \infty).$ Applying the change of variables $D=\mu/T_c$ 
\begin{equation}
    f_{T_c}(y) = f_{|D|}\left(\frac{\mu}{y}\right) \left\lvert -\frac{\mu}{y^2} \right\rvert = \frac{\mu \rho}{v y^2}, \;\;y \in \left[\frac{\mu \rho}{v}, \infty \right).
\end{equation}


\subsection{Proof of Expectation of the fading term}\label{proof-expectation-fading}
\begin{proof}
    The density of a general Nakagami fading environment is defined, with shape parameter $m$,  scale parameter $\Omega$, is defined as  
    $f_K(k) = \frac{1}{\Gamma(m) \Omega^m} k^{m-1} e^{-k/ \Omega}, k > 0$. 

For the performance evaluation in the 2D setting:  
\begin{align}
\mathbb{E}\left[K^{2/ \sigma} \right] &=  \frac{1}{\Gamma(m) \Omega^m} \int_{0}^{\infty} k^{2/ \sigma} k^{m-1} e^{-k/ \Omega} \dd k \notag
    \\
    &= \frac{\Omega^{2/ \sigma}}{\Gamma(m)} \Gamma\left( m + \frac{2}{\sigma} \right).
\end{align}
%
For the Rayleigh fading scenario: $K \stackrel{d}{\sim} \mbox{Exp}(1), f_K(k) {=} e^{-k}, $ for $ k {\geq} 0 $, we have  
    $\mathbb{E}\left[ K^{\frac{2}{\sigma}} \right] {=} \int_{0}^{\infty} k^{\frac{2}{\sigma}} e^{-k} \dd k 
    = \Gamma\left( 1 + \frac{2}{\sigma} \right)$.

For the performance evaluation in the 3D setting, with a general Nakagami distribution on the fading, similar to the 2D setting, we have
\begin{align}
\mathbb{E}\left[ K^{\frac{3}{\sigma}} \right] &=  \frac{1}{\Gamma(m) \Omega^m} \int_{0}^{\infty} k^{\frac{3}{\sigma}} k^{m-1} e^{-k/ \Omega} \dd k \notag\\
&= \frac{\Omega^{\frac{3}{\sigma}}}{\Gamma(m)} \Gamma\left( m + \frac{3}{\sigma} \right) .
\end{align}
For Rayleigh fading: 
$K \stackrel{d}{\sim} \mbox{Exp}(1), f_K(k) {=} e^{-k}$, for $k {\geq} 0,$ \text{we get,} 
    $\mathbb{E}\left[ K^{\frac{3}{\sigma}} \right] {=} \int_{0}^{\infty} k^{\frac{3}{\sigma}} e^{-k} \dd k 
    {=} \Gamma\left( 1 + \frac{3}{\sigma} \right)$.

\end{proof}

\subsection{Proof of Lemma~\ref{RI3D}}\label{proof-prop-2}

\begin{proof} 

Analogous to the 2D setting, 
the Doppler angle is rotational invariant, since the inner product of two vectors preserved by any rotation.
%
%
To elaborate more, we rotate the $x{-}y$ plane such that \underbar{$v$} lies on the $x$-axis. To achieve this, We first rotate the vectors along the $z$-axis by the angle $\beta$, i.e., $\underbar{$R$}_1 {=} R_z(-\beta)$. Then rotate the $x{-}y$ plane along $y$-axis by an angle $\alpha$, i.e., $\underbar{$R$}_2 {=} R_y(\alpha)$. These rotation operations in 3D ensure that  \underbar{$v$} lies on the $x{-}y$ plane. 
The 3D rotation matrices $\underbar{$R$}_1 \ \text{and} \ \underbar{$R$}_2$ are:
\begin{equation} \underbar{$R$}_1 {=} 
    \begin{bmatrix}
        \cos \beta\; & \sin \beta\; & 0 \\
        -\sin \beta\; & \cos \beta\; & 0 \\
        0\; & 0\; & 1
    \end{bmatrix},
     \underbar{$R$}_2 {=} 
    \begin{bmatrix}
        \cos \alpha\; & 0\; & \sin \alpha \\
        0\; & 1\; & 0 \\
        -\sin \alpha\; & 0\; & \cos \alpha 
    \end{bmatrix}.
\end{equation}
%
%
The 3D LoS and DoM vectors post-rotation are given by
\begin{equation} 
\underbar{$r$}^\prime {=} 
%
\underbar{$R$}_2\; \underbar{$R$}_1\;  \underbar{$r$} {=} r 
    \begin{bmatrix}
        \cos \alpha \cos \phi \cos (\theta - \beta) + \sin \alpha \sin \phi \\
        \cos \phi \sin (\theta - \beta) \\
        -\sin \alpha \cos \phi \cos (\theta - \beta) + \cos \alpha \sin \phi
    \end{bmatrix},
\end{equation} 
\text{and}
\begin{equation}\label{post-rotation-v}  
 \ \underbar{$v$}^\prime =\underbar{$R$}_2 \cdot \underbar{$R$}_1 \cdot \underbar{$v$} = v
    \begin{bmatrix}
       1 &
        0 &
        0
    \end{bmatrix}^T.
\end{equation}
The cosine of the Doppler angle $\Gamma$ between the new LoS $\underbar{$r$}^\prime$ and DoM vectors $\underbar{$v$}^\prime$  is 
\begin{equation}\label{3d-gamma-re}
\begin{aligned}
    \cos \Gamma  {=}  \frac{\langle \underbar{$r$}^\prime, \underbar{$v$}^\prime\rangle}{rv} {=} \cos \alpha \cos \Phi \cos (\Theta {-} \beta) {+} \sin \alpha \sin \Phi.
\end{aligned}
\end{equation}
The post-rotation Doppler angle in \eqref{3d-gamma-re}  remains the same as the one pre-rotation, derived in \eqref{3d-gamma}. Hence, the 3D rotation of the plane or axis does not affect the distribution of the Doppler angle and thereby the Doppler shift. This proves the isotropic behavior of the 3D statistical analysis. 
\end{proof}


\subsection{Proof of Theorem \ref{theorem-5}}\label{app-3}

\begin{proof}
We have
$\cos \Gamma {=} \cos \Phi \cos \Theta$,
with $\Phi \stackrel{d}{\sim} \frac{1}{2} \cos \phi$, $\Theta \stackrel{d}{\sim} \mbox{U}[0, 2\pi)$. 
Let us assume the change of variable 
    $Z  {=} \cos \Gamma, 
    X {=} \cos \Phi, 
    Y {=} \cos \Theta$,
which implies $Z {=} XY$. Below, we derive the distribution of $Z = \cos \Gamma$. The PDF of $X = \cos \Phi, \Phi \stackrel{d}{\sim} \frac{1}{2}\cos \phi$, for $\phi\in [-\frac{\pi}{2}, \frac{\pi}{2}]$, is derived through transformation of random variables as 
    \begin{equation}
    \begin{aligned}
        f_X(x) {=} 2 f_{\Phi}(\cos^{-1}x) \left| \frac{\dd \phi}{\dd x} \right| {=}  \frac{x}{\sqrt{1 {-} x^2}},
    \end{aligned}
\end{equation}
%
for all $x \in [0,1]$. Next, the PDF of $Y = \cos \Theta$,  where $\Theta \stackrel{d}{\sim} \mbox{U}[0, 2\pi)$, as derived in \eqref{pdf-y-cos-thetha} is
%
    $f_Y(y) = \frac{1}{\pi \sqrt{1 - y^2}},$ 
%
for $y\in [-1, 1]$.
We derive the PDF of $Z {=} \cos \Gamma {=} X Y$ as follows.  
Based on law of total probability, 
\begin{align}
f_Z(z) &= \int f_{Z|X}(z|x)f_X(x) \dd x = \int f_Y(z/x) \left| \frac{\dd y}{\dd z} \right| f_X(x) \dd x\nn\\
&= \int f_Y(z/x) \frac{1}{x}  f_X(x) \dd x.
\end{align}
To determine the limits of integration, we know that $x {\in} [0,1], y {\in} [-1,1]$. Given $y {=} \frac{z}{x}$ require $\frac{z}{x} \in [-1,1] \ \text{or} \ |z| \leq x$ for $f_Y(z/x)$ to be well defined. 
%
Using the domain of the integral, we have
\begin{equation}\label{integral-simulation-result} \begin{aligned}
    f_Z(z) &= \int_{|z|}^{1} f_Y\left( \frac{z}{x} \right) \frac{1}{x} f_X(x) \dd x \\ 
    &= \frac{1}{\pi} \int_{|z|}^{1}  \frac{x}{ \sqrt{x^2 - z^2}} \frac{1}{\sqrt{1 - x^2}} \dd x.
    \end{aligned}
\end{equation}
To solve the integral, for $z \geq 0$, we make the change of variable $x^2 {-} z^2 {=} u$ in~\eqref{integral-simulation-result} to obtain  
\begin{equation}
    f_Z(z) = \frac{1}{\pi} \int_{0}^{1-z^2} \frac{\dd u }{2 \sqrt{u}\sqrt{1-u - z^2}}.
    \label{eq:f1}
\end{equation}
Further substituting, $1-z^2 = b$, the integral modifies to
\begin{align}
f_Z(z) &{=} \frac{1}{2\pi} \int_{0}^{b} \frac{\dd u}{\sqrt{u}\sqrt{b-u}} {=} \frac{1}{\pi} \arcsin(\sqrt{u/b}) \Bigg\vert_{0}^{b} =\frac{1}{2}.
    \label{eq:f2}
\end{align}
%
%
Hence we have $\cos \Gamma\stackrel{d}{\sim} \mbox{U}[-1, 1]$.  Finally, we derive the PDF of the Doppler shift in this case. The Doppler shift is given as 
    $D = \frac{v}{\rho} \cos \Gamma = \frac{v}{\rho} \underbrace{\cos \Phi \cos \Theta}_{Z}$.
%
The PDF of the Doppler shift, $D$, is computed as
\begin{equation}
\begin{aligned}
    f_{D}(x) &{=}  f_Z\left(\frac{x \rho}{v}\right) \left| \frac{\dd z}{\dd x} \right| {=} \frac{\rho }{2 v}, 
\end{aligned}
\end{equation}
for all $x \in \left[- \frac{v}{\rho}, \frac{v}{\rho} \right]$. This completes the proof. 
%
%

%
%

\end{proof}






\subsection{Proof of closest point distribution in 3D}\label{closest-point-3d}

\begin{proof}
The volume element in a 3D setting is, $dV {=} r^2 \dd r \cos \phi  \dd \phi \dd \theta$. The expected number of points within radius $R$ is  calculated as
\begin{align}
    \Lambda(B_R(o)) {=} \lambda \int_{r=0}^{R} \!\!\! r^2 \dd r \int_{\phi = -\pi/2}^{\pi/2} \!\!\! \!\!\! \!\!\!   \cos \phi  \dd \phi \int_{\theta = 0}^{2 \pi} \!\! \dd \theta
    {=} \frac{4 \pi \lambda R^3}{3}.
\end{align}
%
Let $R_o$ be the random variable of the distance to the nearest BS from the UE. The probability that no point lies within distance $r$ from the UE is computed using the void probability of the underlying PPP $\Pcal$ of BSs on the set $B_r(o)$, i.e., 
\[
\mathbb{P} (R_o {>} r) = \P(\Pcal(B_r(o))=0) = e^{ - 4 \pi \lambda R^3/{3}}.
\]
    %
Hence the distribution of the distance of the closest point in a 3D setting is the corresponding CDF given by, $\mathbb{P}(R_o {\leq} r) = 1 {-} e^{ - 4 \pi \lambda r^3/{3}}$ and the probability density is $f_{R_o}(r) {=} \frac{\dd}{\dd r} \mathbb{P}(R_o {\leq} r)   = 4 \pi \lambda r^2 e^{-4 \pi \lambda r^3/ 3}$.
\end{proof}




\subsection{Doppler shift distribution in LoS  dependent motion}\label{dependent-motion-3d} 

It is notable that the result in Theorem \ref{theorem-3} assumes that all the four angles involved as independent. As a proof of concept, we showcase an example of  LoS dependent motion. 
    
Let us consider a motion model where the DoM depends LoS, for e.g., the azimuths of LoS and DoM are parallel or equal, i.e., $\Theta {=} \beta$. This is a special case, where the BS motion is constrained along the plane formed by the LoS and $z$-axis. Denoting the  Doppler angle as  $\Gamma_2$, we have
\begin{equation}\label{3d-gamma-case1}
\begin{aligned}
    \cos \Gamma_2 &{=} \cos \alpha \cos \Phi {+} \sin \alpha \sin \Phi 
    {=} \cos(\Phi {-} \alpha),
    \end{aligned}
\end{equation}
where the angles $\Phi$ and $\alpha$ are non-uniform as in (\ref{phi-alpga}).

\begin{theorem}\label{theorem-4}
The PDF of the Doppler angle $\Gamma_2$ and the corresponding Doppler shift $D$ are given as 
\begin{align}
        f_{\Gamma_2}(\gamma) &=\frac{1}{8} \left[(\pi-|\gamma|)\cos\gamma + \sin|\gamma|\right],
\end{align}
\begin{align}
      f_D(x) {=} \frac{1}{4} \left[\left(\pi{-}\cos^{-1}\left(\frac{x \rho}{v} \right)\right) \!x \!\left(\!1{-}\!\left( \frac{x \rho}{v} \right)^2\right)^{-\half}\!\!\!\!\!{+} \frac{v}{\rho}\right], 
      \label{Dep-DS-pdf}
\end{align}
for $\gamma\in [-\pi,\pi]$ and $x \in \left[-v/{\rho},v/{\rho} \right]$.
\end{theorem}

\begin{proof}
The  cosine of the Doppler angle in this case comes out to be  $\cos \Gamma {=} \cos(\Phi {-} \alpha)$ (from \eqref{3d-gamma-case1}).
%
%
Let, $\Gamma {=} \vert\Phi {-} \alpha\vert, \Gamma {\in} [-\pi, \pi]$. Then the PDF of $\Gamma$ is
\begin{equation}
\begin{aligned}
   f_\Gamma(\gamma) &=  \int_{L(\gamma)} f_{\Phi} (\phi) f_{\alpha}(\gamma -\phi) \dd \phi.
\end{aligned}
\label{eq:Lgamma}
\end{equation}
The PDF is determined by integrating over a suitable range of values of $\phi$ shown below, 
\[
L(\gamma):=\left[-\frac{\pi}{2}\vee \left(\gamma-\frac{\pi}{2}\right), \frac{\pi}{2}\wedge \left(\gamma+\frac{\pi}{2}\right)\right] \equiv [\ell_\gamma, r_\gamma],
\] 
for given $\gamma$. For $\gamma<0$, we have $[\ell_\gamma, r_\gamma]= \left[-\frac{\pi}{2}, \gamma+\frac{\pi}{2}\right]$ and for $\gamma\geq 0$, we have $[\ell_\gamma, r_\gamma]= \left[\gamma-\frac{\pi}{2}, \frac{\pi}{2}\right]$. In the first case, with $\gamma<0$, the integral in (\ref{eq:Lgamma}) is equal to 
\begin{align}
\frac{1}{4} \int_{-\pi/2}^{\gamma+\pi/2}\!\!\!\!\!\!\! \cos \phi \cos (\gamma {-} \phi)  \,\dd \phi 
&{=} \frac{1}{8} \int_{-\pi/2}^{\gamma+\pi/2} \!\!\!\!\! [\cos\gamma {+} \cos(2\phi {-} \gamma)]\, \dd \phi \nn \\
%
&{=} \frac{1}{8} \left[(\gamma{+}\pi)\cos\gamma {-} \sin\gamma\right].
\end{align}
In the other case, with $\gamma\geq 0$, the integral in (\ref{eq:Lgamma}) can similarly be computed as
\begin{align}
\frac{1}{4} \int_{\gamma-\pi/2}^{\pi/2}\hspace{-0.1in} \cos \phi \cos (\gamma {-} \phi)  \,\dd \phi
&= \frac{1}{8} \left[(\pi{-}\gamma)\cos\gamma + \sin\gamma\right].
\end{align}
Thus the density of $\Gamma$ turns out to be
\begin{align}
 f_\Gamma(\gamma) &= \frac{1}{8} \left[(\pi-|\gamma|)\cos\gamma + \sin|\gamma|\right], 
\end{align}
for $\gamma\in [-\pi,\pi]$.
The probability density of $\cos\Gamma:=Z$ is computed using the substitution $\gamma=\arccos z$ and $\sin|\gamma|= \sqrt{1-z^2}$ to obtain
\begin{align}
f_{Z}(z)&= 2\times \frac{1}{8} \left[(\pi-|\arccos z|) z+ \sqrt{1-z^2}\right] \left\vert\frac{-1}{\sqrt{1-z^2}}\right\vert\nn\\
&= \frac{1}{4} \left[(\pi-\arccos z)\frac{z}{\sqrt{1-z^2}}+1\right],
\end{align}
for $z\in [-1,1]$.
Then, the probability density of the Doppler shift $D   = \frac{v}{\rho}\cos\Gamma$ is $f_{D}(x)$ is  
\begin{equation}\label{3d-case-1}
\begin{aligned}
    \frac{v}{4\rho} \left[\left(\pi{-}\arccos\left(\frac{x \rho}{v} \right)\right) \frac{x \rho}{v} \left(1{-}\left( \frac{x \rho}{v} \right)^2\right)^{-\half}{+}1\right],
    \end{aligned}
\end{equation}
for $x \in \left[ -\frac{v}{\rho}, \frac{v}{\rho} \right]$. The expression in \eqref{3d-case-1} shows the density of the Doppler shift as a function of non-uniform angles.
\end{proof}


\bibliographystyle{IEEEtran}
\bibliography{IEEEabrv,ref}

\end{document}